\documentclass[journal]{IEEEtran}
\ifCLASSINFOpdf
\usepackage[pdftex]{graphicx}
\graphicspath{{../pdf/}{../jpeg/}}
\DeclareGraphicsExtensions{.pdf,.jpeg,.png}
\else
\usepackage[dvips]{graphicx}
\graphicspath{{../eps/}}
\DeclareGraphicsExtensions{.eps}
\fi
\usepackage{graphicx}
\usepackage{graphics}
\usepackage{epsfig}
\usepackage{epstopdf}
\usepackage{stfloats}
\usepackage[cmex10]{amsmath}
\usepackage{algorithmic}
\usepackage{array}
\usepackage{mdwmath}
\usepackage{dsfont}
\usepackage{mdwtab}
\usepackage{graphicx}
\usepackage{subfigure}
\usepackage{color}
\usepackage{amsfonts,amssymb}
\usepackage{hyperref} 
\usepackage{lineno} 
\usepackage{multirow}
\usepackage{diagbox}
\usepackage{caption}
\usepackage{bbding}
\usepackage{amsfonts,amsthm,array} 
\usepackage[ruled]{algorithm2e}
\usepackage{makecell}

\usepackage{cuted}

\newtheorem{lemma}{Lemma}

\usepackage{makecell}

\usepackage{bbding}

\begin{document}
	
	
	\title{Proactive Eavesdropping in Relay Systems via Trajectory and Power Optimization	\thanks{Manuscript received.}}
	
	\author{Qian~Dan, 
		Hongjiang~Lei,
		Ki-Hong~Park, 
		Weijia~Lei,
		Gaofeng~Pan
		\thanks{This work was supported by the National Natural Science Foundation of China under Grant 61971080. (Corresponding author: \textit{Hongjiang~Lei}).}
		\thanks{Qian~Dan is with the School of Communications and Information Engineering, Chongqing University of Posts and Telecommunications, Chongqing 400065, China and also with the School of Computer Science, Jiangxi University of Chinese Medicine, Nanchang 330004, China (e-mail: {danqian@jxutcm.edu.cn}).}
		\thanks{Hongjiang~Lei and Weijia~Lei are with the School of Communications and Information Engineering, Chongqing University of Posts and Telecommunications, Chongqing 400065, China and also is Chongqing Key Lab of Mobile Communications Technology, Chongqing 400065, China (e-mail: leihj@cqupt.edu.cn; leiwj@cqupt.edu.cn).}	
		\thanks{K.-H. Park is with CEMSE Division, King Abdullah University of Science and Technology (KAUST), Thuwal 23955-6900, Saudi Arabia (e-mail: kihong.park@kaust.edu.sa).}
		\thanks{Gaofeng~Pan is with the School of Cyberspace Science and Technology, Beijing Institute of Technology, Beijing 100081, China (e-mail: gaofeng.pan.cn@ieee.org).}
		
	}
	
	\maketitle
	
	\begin{abstract}
		Wireless relays can effectively extend the transmission range of information. However, if relay technology is utilized unlawfully, it can amplify potential harm. 
		Effectively surveilling illegitimate relay links poses a challenging problem. 
		Unmanned aerial vehicles (UAVs) can proactively surveil wireless relay systems due to their flexible mobility. 
		This work focuses on maximizing the  eavesdropping rate (ER) of UAVs by jointly optimizing the trajectory and jamming power. 
		To address this challenge, we propose a new iterative algorithm based on block coordinate descent and successive convex approximation technologies. 
		Simulation results demonstrate that the proposed algorithm significantly enhances the ER through trajectory and jamming power optimization. 
	\end{abstract}
	
	
	
	\section{Introduction}
	\label{sec:introduction}
	
	\subsection{Background and Related Works}
	\label{sec:Background}   
	
	The unmanned aerial vehicle (UAV) possesses the advantages of compact size, high speed, and versatility, enabling it to carry relevant equipment for various scenarios such as rescue operations and aerial photography \cite{ZengY2016CM, Hayat2016COMST}, among others. 
	The combination of UAVs and the Internet of Things (IoT) can realize real-time acquisition, processing, and transmission of information and improve UAVs' work efficiency and application value.
	The emergence of UAVs has significantly expanded the possibilities for wireless communication scenarios, presenting both broad prospects and challenges in UAV communication \cite{Bander2020JNCA, Sahoo2022IOT}. 
	In scenarios where disasters occur and ground base stations (BSs) are damaged, UAVs can be utilized to restore communication swiftly \cite{ZhaoN2019WC}. 
	Temporary establishment of hotspots by UAVs in areas experiencing high traffic volume can effectively enhance users' communication quality and experience.

	There is growing research on UAV communication, which has become a prominent area in wireless communication. 
	In Ref. \cite{YangH2022IoT}, considering the energy budget and quality of service of the IoT devices, a data dissemination problem was formulated to minimize the maximum completion time by optimizing the UAVs' trajectory and transmitting power consumption, and allocated time for the IoT devices.
	To obtain the tradeoff between the number of UAVs and the coverage rate, in \cite{ZhangC2021TCOM}, the weighted summation of the number of UAVs and the reciprocal coverage rate was minimized by jointly optimizing the positions of UAVs, the user association with UAVs and ground users (GUs), and the frequency bandwidth allocation.
	In Ref. \cite{LiY2021TWC}, the authors maximized the energy efficiency (EE) of an aerial downlink non-orthogonal multiple access (NOMA) network by jointly designing the UAV's transmit power and trajectory and the user scheduling. 
	Ref. \cite{LiuT2021TGCN} considered a cooperative system with multiple aerial relays. The minimum throughput of all the users was maximized by jointly optimizing the relays' three-dimensional (3D) trajectories and transmit power and the association between relays and transmitters.
	In Ref. \cite{LiangY2022TVT}, the authors proposed a new scheme, named time-slots pairing, for the fix-wing UAV-aided two-way cooperative systems and maximized the sum rate by joint optimizing the trajectory design and link scheduling, and power allocation.
	Ref. \cite{FengT2022TWC} investigated the data collection problem in the UAV-assisted wireless sensor network. 
	Focusing on the practical aspects of flight speed and user power constraints, they aimed to maximize the average throughput in delay-tolerant scenarios and minimize the outage probability (OP) in delay-sensitive scenarios through joint optimization of the UAV’s trajectory design and the GUs’ transmit power allocation over time. 
	In Ref. \cite{Samir2020TWC}, a UAV was utilized to collect data from IoT devices, and the number of served IoT devices was maximized by jointly optimizing the trajectory and user scheduling and allocated bandwidth. 
	In Ref. \cite{Dabiri2020CL}, the authors considered a UAV-aided free space optical decode-and-forward (DF) cooperative system, and the end-to-end OP was minimized by jointly designing the divergence angles and the elevation angle of source and destination.

	Because of the line-of-sight (LoS) channel conditions, UAV-aided communication systems are more susceptible to be wiretapped by malicious eavesdroppers than terrestrial communication systems. Thus, the security of UAV communication systems has become a prominent research area in academia. 
	In Ref. \cite{TangG2022SJ}, wireless power transfer technology was utilized in UAV communication systems and the secrecy sum rate (SSR) was maximized by jointly optimizing the wireless charging time, UAV’s trajectory and transmit power ensuring limitations on UAV flight speed, battery capacity, and wireless power supply. 
	In many works, such as Refs. \cite{ZhongC2019CL}-\cite{LeiH2023IoT}, the cooperative jamming method was utilized wherein one UAV was used as a friendly jammer to send artificial noise (AN) to fight with eavesdroppers to protect the communication between the aerial BS and GUs. The average secrecy rate (ASR) was maximized by jointly optimizing the UAVs' trajectory and communication/jamming power allocation in \cite{ZhongC2019CL}.
	The secrecy performance of dual UAV-assisted mobile edge computing (MEC) systems with time division multiple access and NOMA schemes was investigated in \cite{XuY2021TCOM}. The minimum secure computing capacity (SCC) was maximized by jointly optimizing the UAVs' trajectories, time/power allocation, and task offloading ratio (TOR).
	Considering the location-uncertainty eavesdropper in Ref. \cite{LeiH2023IoT} and the ASR of the underlay system was maximized by optimizing user scheduling and the UAVs' trajectory and transmit power. 
	In \cite{ZhouY2022TVT}, with the help of the aerial friendly jammer, the sum of the intercept probability (IP) and the maximum OP among all terrestrial links was minimized by jointly optimizing the transmitting power at the ground BS and the aerial jammer, as well as the location of the jammer. 
	The minimum secure offload rate was maximized by jointly optimizing the trajectories of UAVs, transmit power, and allocated computing frequency. 
	Moreover, the authors of Ref. \cite{HanD2020ChinaCom}  maximized the secrecy rate (SR) of the UAV-assisted MEC system by jointly optimizing the offloaded task and the position and transmit power of the legitimate UAV.
	In Ref. \cite{DuoB2021ChinaCom}, the ASR of the downlink and uplink communications were maximized, respectively, by jointly optimizing the UAV trajectory and transmit power on both the UAV and GUs.
	Ref. \cite{Savkin2020WC} considered the scenarios with the static or mobile eavesdroppers; the energy consumption was minimized by designing the UAV's trajectory while ensuring communication security.  
	In Ref. \cite{WuJ2023TVT}, the UAV was used as an integrated sensing and communication BS and the terrestrial eavesdropper was assumed to work with a fixed known trajectory. 
	The location of the legitimate user was predicted by the proposed extended Kalman filtering method and  the real-time SR was maximized by optimizing the UAV's trajectory.  
	
	The supervision of unauthorized wireless communication links has always been a focal point of official attention. 
	\textit{To the best of the authors' knowledge, there are three main difference between the physical layer security (PLS) and proactive eavesdropping (PE) (also called as wireless surveillance): 
		1) \textbf{Application scenarios}: 
		In PLS scenarios, communication users are assumed to be rightful, and the eavesdroppers are illegitimate, and various schemes are designed to prevent information leakage to the malicious or potential eavesdroppers \cite{TangG2022SJ}-\cite{ZhongC2019CL}.
		In PE scenarios, it is assumed that the eavesdroppers are legitimate and approved by government agencies and that the communication users are suspicious, such as criminals or terrorists who may jeopardize public safety \cite{ZengY2016TSP}-\cite{XuJ2017WC}.	
		2) \textbf{Performance metric}: 
		In PLS scenarios, the fundamental performance metric is the achievable SR, defined as the (non-negative) difference between the achievable rate of the communication link ($R_D$) and the achievable rate of the eavesdropping link ($R_E$). Based on SR, there are two main performance metrics: ergodic SR and secrecy outage probability, which denote the statistical average or time average of SR and the probability that the SR falls below a target rate \cite{TangG2022SJ}-\cite{ZhongC2019CL}.
		In PE scenarios, the relationship between $R_D$ and $R_E$ is crucial. When $R_E \ge R_D$, surveillance is successful, and $R_D$ is called the eavesdropping rate (ER). Otherwise, surveillance fails, and the ER is equal to zero \cite{ZengY2016TSP}-\cite{XuJ2017WC}.
		There are two main performance metrics: average ER (AER) and successful eavesdropping probability (SEP), also called as eavesdropping non-interrupt probability \cite{WuZ2023OPENJ},\cite{HuG2020CLAF}. 	
		3) \textbf{Motivation}: The motivation of schemes in PLS is to  make the ASR as larger as possible while  the aim in wireless surveillance is make $R_D$ as large as possible under the $R_E \ge R_D$ constraint \cite{ZengY2016TSP}-\cite{XuJ2017WC}.}
	
	Ref. \cite{XuD2022TWC} investigated PE by a single eavesdropper spoofing relays over multiple suspicious links, optimized the jamming power and assisting power to achieve maximum success rate and AER while considering quality of service constraints for the suspect links. 
	In \cite{XuD2023TIFS}, the authors maximized the ER by optimizing power allocation when an eavesdropper communicates with multiple suspicious users, while these users prevented PE using specific codes. 
	Ref. \cite{ZhangH2020TWC} explored the use of multi-antenna jamming for PE and characterized achievable regions for the ER through optimization of interference transmission covariance matrix performed by a legitimate monitor. 
	In Ref. \cite{ChenJ2022TVT}, the authors studied spoofing relay PE in massive multiple-input multiple-output orthogonal frequency division multiplexing systems and proposed a deep reinforcement learning (DRL)-based algorithm to solve the optimization problem for maximizing the ER. 
	In \cite{Feizi2020TCOM}, the eavesdropping non-disruption probability was maximized by the waveform design of the transmit and receive beamformers in full-duplex (FD) multi-antenna systems. 
	
	In many scenarios, UAVs also were used as aerial eavesdroppers to surveil suspicious wireless communication.
	In Ref. \cite{LuH2019TVT}, a terrestrial FD monitor surveilled and interfered with an air-to-ground (A2G) link, optimizing the jamming power to maximize the ER. 
	Similarly, in Ref. \cite{HuangM2021WCNCW}, an aerial FD UAV surveilled and interfered with an air-to-air link by jointly designing a two-dimensional (2D) trajectory and jamming power to maximize the ER. 
	Furthermore, Ref. \cite{GuoD2023TMC} demonstrated the use of multiple UAVs for PE on multiple A2G links, where the trajectory and jamming power were jointly optimized to improve performance. 
	Additionally, in Ref. \cite{DanQ2022phycom}, an amplify-and-forward (AF) FD UAV played dual roles as a relay and jammer while surveilling a ground-to-ground link through joint optimization of amplification coefficient, power allocation ratio, and trajectory to maximize the ER.
	
		Relay technology is an effective method of extending the communication range and improving the signal-to-noise ratio (SNR) at the receivers. 
		This method may also be used by malicious users (e.g., criminals, terrorists, and business spies) to commit crimes, jeopardize public safety, invade the secret databases of other companies, and so on, thus imposing new challenges in public security, as stated in \cite{XuJ2017WC}. 		
		In the illegal cooperative communication system, the relay node not only extends the communication range or improves the SNR of the illicit receiver but also provides better monitoring conditions for eavesdroppers. 
		Thus, PE in cooperative systems has been investigated in many works, such as \cite{JiangX2017SPL}-\cite{HuG2020CLAF}.
	Specifically, the relays in these works were designed to serve suspicious communications or be utilized as eavesdroppers.
	The PE problem in a two-hop DF cooperative system was studied in Ref. \cite{JiangX2017SPL}, where the ER was maximized by jointly optimizing beam formation and power allocation.
	In Ref. \cite{HuG2021SJ}, the authors proposed two PE strategies in the cooperative system with a two-way AF relay. The ER was maximized by the joint design the jamming power and beamforming vectors. 
	
	Compared with terrestrial eavesdroppers, aerial eavesdroppers have the characteristics of fast mobility and flexible application, which can greatly improve surveillance performance by jointly designing the trajectory and the jamming power. 
		Thus, UAV-aided PE in cooperative networks has been investigated in many works, such as \cite{HuG2022SPL}-\cite{HuG2020CLAF}.
	In Refs. \cite{HuG2022SPL} and \cite{HuG2020CL}, the authors investigated the PE problem in the cooperative system with an aerial DF relay. 
	Two jamming eavesdropping schemes were proposed and the analytical expressions for the eavesdropping outage probability were derived in \cite{HuG2022SPL}. Moreover, the average eavesdropping throughput and proactive ER were maximized by optimizing jamming power in Refs. \cite{HuG2022SPL} and \cite{HuG2020CL}, respectively. 
	In Ref. \cite{WuZ2023OPENJ}, the authors investigated PE in integrated satellite-terrestrial relay networks and derived the analytical expressions for the exact and asymptotic eavesdropping non-outage probability. 
	In Ref. \cite{HuG2020CLAF}, the authors studied PE in the cooperative system with multiple AF relays. Multiple hovering UAVs were utilized as jammers or eavesdroppers, and the relationship between the success probability of eavesdropping, the flight height of UAVs, and the number of UAVs was analyzed. 
	It should be noted that the relays in these works are to aid suspicious users and not for eavesdropping. 
	The PE scenarios with multi-antenna relays and a multi-antenna jammer were investigated in Ref. \cite{Moon2018TWC}, and the ER was maximized by jointly optimizing the forward precoding matrix, the covariance matrix of the jamming antennas, and the reception vector of the central monitor. 
	In Ref. \cite{GeY2022IOT}, the authors investigated wireless surveillance in the cooperative overlay cognitive radio networks. 
	The cognitive transmitter worked in FD mode and forwarded the receiver suspicious signals to the monitor.
	Moreover, both delay-tolerant and delay-sensitive scenarios were considered, wherein  
	the cognitive transmitter worked as a spoofing relay in AF FD mode or as a jammer, respectively.
	The EE was maximized by jointly optimizing the AF forward matrix and secondary transmitter precoding vectors, as well as the receiver combination vectors.
	It should be noted that, the relays in Refs. \cite{Moon2018TWC} and \cite{GeY2022IOT} were utilized for surveillance.
	{Table \ref{table1} outlines recent literature related to PE.}
	
	\begin{table*}[tb]
		\setcounter{table}{0}
		\centering
		{
			\footnotesize
			\caption{\textit{Related Works to PLS and Proactive Eavesdropping.}}
			\label{table1}
			\centering
			\scalebox{1}{
				\begin{tabular}{c|c|c|c|c|c|c}
					\Xhline{1.2pt}
					{Reference} &{UAVs}&{\makecell[c]{PLS/PE}} &{\makecell[c]{Single-hop/ \\ Dual-hop}}  &{\makecell[c]{Performance metric}}&{\makecell[c]{Research method}}&{\makecell[c]{Main parameters}}\\
					\hline
					{\cite{TangG2022SJ}} 	 &{\checkmark}&{PLS}			&{Single}	&{Sum of SR}	&{Convex optimization}	&{\makecell[c]{Trajectory, wireless charging \\duration, transmit power}}\\
					\hline
					{\cite{ZhongC2019CL}} 	 &{\checkmark}&{PLS}			&{Single}	&{ASR}	&{Convex optimization}	&{\makecell[c]{ Trajectory, transmit power}}\\
					\hline
					{\cite{XuY2021TCOM}} 	 &{\checkmark}&{PLS}			&{Single}	&{\makecell[c]{SCC}}	&{Convex optimization}	&{\makecell[c]{ Trajectory, time/power allocation, TOR}}\\
					\hline
					{\cite{LeiH2023IoT}} 	 &{\checkmark}&{PLS}			&{Single}	&{ASR}	&{Convex optimization}	&{\makecell[c]{Trajectory, user scheduling}}\\
					\hline
					{\cite{ZhouY2022TVT}} 	 &{\checkmark}&{PLS}			&{Single}	&{\makecell[c]{OP, IP} }	&{Convex optimization}	&{\makecell[c]{BS transmit power,\\ UAV  location,\\ jamming power}}\\
					\hline
					{\cite{HanD2020ChinaCom}} 	 &{\checkmark}&{PLS}		&{Single}	&{SR}	&{Convex optimization}	&{\makecell[c]{Trajectory,transmit power}}\\
					\hline
					{\cite{DuoB2021ChinaCom}} 	&{\checkmark}&{PLS} 		&{Single}	&{ASR}	&{Convex optimization}	&{\makecell[c]{UAV positions, transmit 
							power, TOR}}\\
					\hline
					{\cite{WuJ2023TVT}} 	 &{\checkmark}&{PLS}			&{Single}	&{SR}	&{Convex optimization}	&{\makecell[c]{Trajectory}}\\
					\hline
					{\cite{ZengY2016TSP}} 	 &{}&{PE}		    &{Dual}	&{ASR}	&{Convex optimization}	&{\makecell[c]{ Power
							splitting ratios, \\relay precoding matrix}}\\
					\hline
					{\cite{XuJ2017TWC}} 	 &{}&{PE}			&{Single}	&{ASR}	&{Convex optimization}	&{\makecell[c]{Jamming power allocation}}\\
					\hline
					{\cite{XuD2022TWC}} 	 &{}&{PE}			&{Single}	&{ER}	&{Convex optimization}	&{\makecell[c]{Eavesdropping strategy,\\ jamming power}}\\
					\hline
					{\cite{XuD2023TIFS}} 	&{}&{PE}			&{Single}	&{ER}	&{Convex optimization}	&{Jamming power distribution}\\
					\hline
					{\cite{ZhangH2020TWC}}	&{}&{PE}			&{Single}	&{ER}	&{Convex optimization}	&{Jamming covariance matrix}\\
					\hline
					{\cite{ChenJ2022TVT}} 	&{}&{PE}			&{Single}	&{ER}	&{DRL Optimization}		&{Power allocation}\\
					\hline
					{\cite{Feizi2020TCOM}}  &{}&{PE}			&{Dual}		&{SEP}	&{Convex optimization}	&{Transmit and receive vectors}\\
					\hline
					{\cite{LuH2019TVT}}		&{\checkmark}&{PE}	&{Single}	&{Sum of ERs}	&{Convex optimization}	&{Jamming power}\\
					\hline
					{\cite{HuangM2021WCNCW}}&{\checkmark}&{PE} &{Single}	&{AER}	&{Convex optimization}  &{\makecell[c]{Jamming power, trajectory}} \\
					\hline
					{\cite{GuoD2023TMC}}	&{\checkmark}&{PE}	&{Single}	&{\makecell[c]{AER}}	&{DRL Optimization}	&{\makecell[c]{Jamming power, trajectory}}\\			
					\hline
					{\cite{DanQ2022phycom}} &{\checkmark}&{PE} &{Single}	&{ER}	&{Convex optimization}	&{\makecell[c]{Amplification factor, power \\allocation, trajectory}}\\
					\hline
					{\cite{JiangX2017SPL}}		&{}&{PE}			&{Dual}			&{ER}	&{Convex optimization}	&{\makecell[c]{Jamming beamformer and power}}\\
					\hline
					{\cite{HuG2021SJ}}			&{}&{PE}			&{Dual}	&{ER}	&{Convex optimization}	&{\makecell[c]{Jamming beamformer and power}}\\
					\hline
					{\cite{HuG2022SPL}}		&{\checkmark}&{PE}		&{Dual}	&{ER}	&{Convex optimization}	&{Jamming power}\\
					\hline
					{\cite{HuG2020CL}}		&{\checkmark}&{PE}&{Dual}	&{ER}	&{Convex optimization}	&{Jamming power}\\
					\hline
					{\cite{WuZ2023OPENJ}}	&{}&{PE}				&{Dual}	&{SEP}	&{Performance analysis}	&{Eavesdropping mode}\\
					\hline
					{\cite{HuG2020CLAF}}	&{\checkmark}&{PE}	&{Dual}	&{SEP}	&{Performance analysis}	&{Number of UAVs, height}\\
					\hline
					{\cite{Moon2018TWC}}	&{}	&{PE}			&{Dual}	&{ER}	&{Convex optimization}	&{\makecell[c]{Precoding vector, receiving vector,\\covariance matrix of interference }}\\
					\hline
					{\cite{GeY2022IOT}}			&{}	&{PE}		&{Dual}	&{EE}	&{Convex optimization}&{Precoding vector, receiving vector}\\
					\hline
					{Our Work}					&{\checkmark}&{PE}&{Dual}	&{AER}	&{Convex optimization}	&{\makecell[c]{Jamming power, trajectory}}\\
					\Xhline{1.2pt}
				\end{tabular}
			}	
		}
		
	\end{table*}

	\subsection{Motivation and Contributions}
	\label{sec:Motivation}
	
	The previous work has indicated that PE can be achieved by transmitting jamming signals. However, these excellent works do not answer the following issue: 
	\textit{
		How do we design the flight trajectory and jamming power of the aerial legitimate eavesdropper to successfully eavesdrop on the illegal relay links? 
		What is the optimal eavesdropping performance of the PE scenarios with a legitimate aerial eavesdropper?
	}
	Therefore, this work focuses on maximizing the ER of UAVs by jointly optimizing the trajectory and jamming power.
	The main contributions of this paper are listed as follows: 
	
	\begin{enumerate}
		\item We investigate PE in the cooperative system with an AF relay located on the top of the building. 
		The aerial FD eavesdropper not only tries to wiretap the suspicious information from the source and the relay but also sends jamming noise to reduce the reception quality of the illegitimate destination. 
		The AER is optimized by designing the 2D/3D trajectory and transmission power of AN, while the condition for successful eavesdropping is guaranteed.
		To solve this complex problem, we decompose the difficult PE problem into three sub-problems, power optimization, horizontal position optimization, and vertical height optimization by block coordinate descent (BCD) method, and solve the sub-problem by successive convex approximation (SCA). 
		
		\item Relative to \cite{HuangM2021WCNCW, GuoD2023TMC,  DanQ2022phycom} in which the UAV was utilized to eavesdrop on the single-hop terrestrial wireless link, a cooperative system with an AF relay and LoS links are considered in this work. Relative to \cite{HuG2022SPL, HuG2020CL} in which UAVs were assumed to forward signals to suspicious users, a UAV in this work is utilized to eavesdrop on the cooperation systems with LoS links. Moreover, the UAV's 3D trajectory is also jointly designed with the jamming power to enhance the AER.
		
		\item Relative to \cite{HuG2020CLAF} in which PE in the cooperative system with the AF relay was investigated and the analytical expressions for the exact and asymptotic eavesdropping non-outage probability were derived, this work optimizes the ER by joint designing the UAV's 2D/3D trajectory and jamming power. Although the PE problem in cooperative systems was investigated in Refs. \cite{Moon2018TWC} and \cite{GeY2022IOT} wherein the ER was maximized by designing the precoding matrix at the relay, we maximize the ER by joint designing the UAV's 2D/3D trajectory and jamming power.
		
	\end{enumerate}
	
	\subsection{Organization}
	
	The remainder of this paper is organized as follows. 
	Sect. \ref{sec:SystemModel} introduces the illegitimate cooperative system model and Sect. \ref{sec:Problem1} presents the formulated problem. 
	In Sect. \ref{sec:ProposedAlgorithm1}, an iterative algorithm is proposed to solve the formulated problem based on the SCA and BCD techniques. 
	In Sect. \ref{sec:Problem2}, the formulated problem was extended to the scenarios where the aerial eavesdropper works at a varying altitude and the 3D trajectory and jamming power of eavesdropper are jointly optimized.
	The simulation results are analyzed in Sect. \ref{sec:Simulation}. 
	Finally, the conclusions of this work are given in Sect. \ref{sec:Conclusions}.
	
	\section{System Model}
	\label{sec:SystemModel}
	
	\begin{figure}[t]
		\centering		
		\includegraphics[width = 3.8 in]{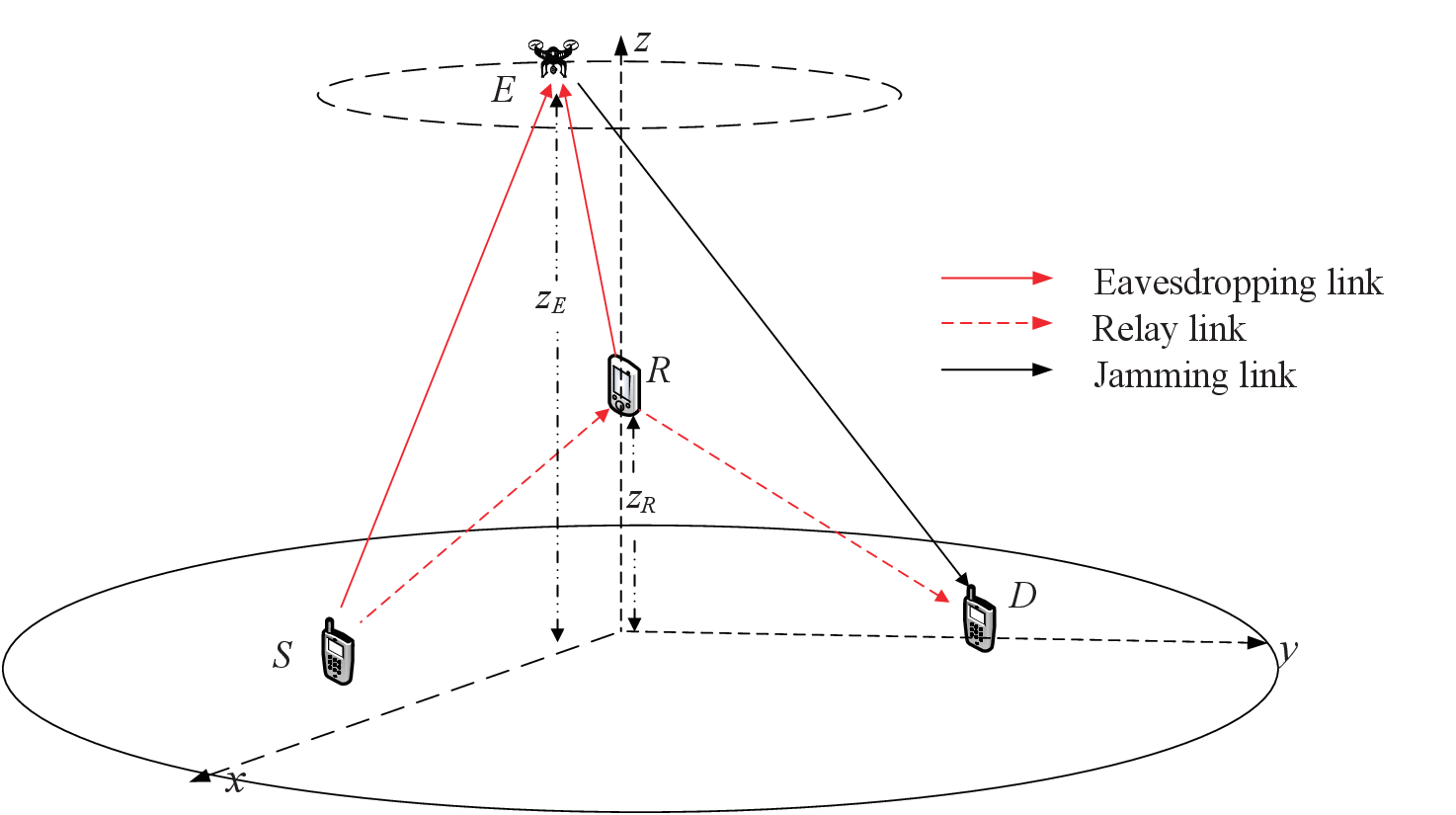}
		\caption{System model consisting of an illegal terrestrial source $\left( S \right)$, a ground destination $\left( D \right)$, an AF relay $\left( R \right)$, and an aerial FD eavesdropper $\left( E \right)$.}
		\label{fig_model}
	\end{figure}

	\begin{table}[tb]
			\footnotesize
			\setcounter{table}{1}
			\caption{\textit{List of symbol Notations.}}
			\begin{center}
				\begin{tabular}{c| c }
					\Xhline{1.2pt}
					\textbf{Notation}   	& \textbf{Description}								\\
					\hline
					${\mathbf{q}}_I$             & Initial horizontal location of $E$ 							\\
					\hline
					${\mathbf{q}}_F$               & Final horizontal location of $E$ 							\\
					\hline	
					${\mathbf{q}}_S $         & Horizontal location of $S$ 					\\
					\hline
					${\mathbf{q}}_R$     & Horizontal location of $R$ 					\\
					\hline
					${\mathbf{q}}_D$			& Horizontal location of $D$ 					\\
					\hline
					${\mathbf{q}}_E$			& Horizontal location of $E$ 					\\
					\hline
					${{H}}_I$             & Initial vertical location of $E$ 							\\
					\hline
					${{H}}_F$               & Final vertical location of $E$ 							\\
					\hline
					$z_{R}$	                & Vertical height of $R$	\\
					\hline
					$z_{E}$	 					  & Vertical height of $E$								\\
					\hline
					$v^{\max}_{xy}$	                & Maximum horizontal speed	\\
					\hline
					$v^{\max}_{z}$					  & Maximum vertical speed								\\
					\hline
					$a^{\max}_{xy}$               & Maximum horizontal acceleration  					\\
					\hline
					$a^{\max}_{z}$	             & Maximum vertical acceleration  \\ 
					\hline
					$H_{max}$					& Maximum flight altitude				\\ 
					\hline
					$H_{min}$				& Minimum flying altitude					\\ 
					\hline
					$ \sigma^2$				& The noise power			\\ 
					\hline
					$\beta _0$				& The channel power gain at the reference distance	 			\\ 
					\hline
					$P_S$ 					& Transmit power of $S$				\\ 
					\hline
					$P_R$ 					& Transmit power of $R$ 			\\ 
					\hline
					$\varepsilon$   & Algorithm convergence precision\\
					\hline
					${P_0}$&   Blade profile power\\
					\hline
					${P_i}$&    Induced power\\
					\hline
					${v_0}$&   The mean rotor induced velocity\\
					\hline
					${U_{\mathrm{tip}}}$&     Tip speed of the rotor blade \\
					\hline
					${d_0}$&        Fuselage drag ratio\\
					\hline
					$\rho$ &         Air density\\
					\hline
					$s$ &           Rotor solidity\\
					\hline
					$A$ &           Rotor disc area\\
					\hline
					$P_{\mathrm{hor}}^{\mathrm{ave} }$& Average horizontal flight power\\
					\hline
					$P_{\mathrm{ver}}^{\mathrm{ave} }$& Average vertical flight power\\	
					\hline
					$W$ & Aircraft weight \\
					\Xhline{1.2pt}
				\end{tabular}
			\end{center}
			\label{table2}
	\end{table}
	
	As shown in Fig. \ref{fig_model}, we consider a PE scenario wherein the illegitimate terrestrial source $\left( S \right)$ transmits suspicious messages to the GU $\left( D \right)$ with an AF FD relay $\left( R \right)$ located on the top of the building with a height of $z_R$. 
	Due to the obstruction of the building, there is no direct transmission link between $S$ and $D$. 
	It is assumed that there is an adaptive transmission over $S$-$R$ and $R$-$D$ links, and the transmission rate of the two links are affected by the AN sent by $E$ \cite{ZengY2016TSP}-\cite{XuJ2017WC}.
	An aerial FD eavesdropper $\left( E \right)$ not only tries to wiretap the suspicious information from $S$ and $R$ but also sends jamming noise to reduce the reception quality of $D$ to guarantee positive ER. 
	The flight period, $T$, is divided into ${N}$ time slots as ${\delta _t} = \frac{T}{N}$.
	When ${\delta _t}$ is small enough, the position of UAVs at each point can be approximated as a continuous trajectory \cite{WuQ2018TWC}.
	The horizontal coordinates of $S$ and $D$ are expressed as
	${{\mathbf{q}}_{S}} = {\left[ {{x_{S}},{y_{S}}} \right]^T}$ 
	and 
	${{\mathbf{q}}_{D}} = {\left[ {{x_{D}},{y_{D}}} \right]^T}$, 
	respectively. 
	It is assumed that the horizontal coordinate of $R$ is expressed as
	${{\mathbf{q}}_{R}} = {\left[ {{x_{R}},{y_{R}}} \right]^T}$. 
	To avoid collisions, the flight altitude $z_E$ of $E$ should be higher than the altitude of the highest obstacle in the service area \cite{CuiM2018TVT}, \cite{ZhangR2021TWC}.
	The horizontal position of $E$ in each time slot is expressed as 
	${{\mathbf{q}}_E}\left[ n \right] = {\left[ {{x_E}\left[ n \right],{y_E}\left[ n \right]} \right]^T}$, 
	where $n =0, 1, \cdots ,N$.

	Similar to \cite{YangH2022IoT, LiuT2021TGCN, ZhongC2019CL}, it is assumed all the wireless links are LoS links and the channel gain for $S$ to $E$ and $R$ are expressed as
	\begin{subequations}
		\begin{align}
			h_{SE}^2 &= \frac{{{\beta _0}}}{{{{\left\| {{{\mathbf{q}}_E}\left[ n \right] - {{\mathbf{q}}_S}} \right\|}^2} + {z_E^2}}}, \label{hse}\\
			h_{SR}^2 &= \frac{{{\beta _0}}}{{{{\left\| {{{\mathbf{q}}_R} - {{\mathbf{q}}_S}} \right\|}^2} + {\textcolor{blue}{z}_{R}^2}}}, \label{hsr}
		\end{align}
	\end{subequations}
	where $\beta _0$ denotes the channel power gain at the reference distance. 
	Similarly, the channel gain for $R$ to $E$ and $D$ and the channel gain for $E$ to $D$ are expressed as
	\begin{subequations}
		\begin{align}
			h_{RE}^2 & = \frac{{{\beta _0}}}{{{{\left\| {{{\mathbf{q}}_E}\left[ n \right] - {{\mathbf{q}}_R}} \right\|}^2} + {{\left\| {{z_E} - {z_R}} \right\|}^2}}}, \label{hre}\\
			h_{RD}^2 & = \frac{{{\beta _0}}}{{{{\left\| {{{\mathbf{q}}_R} - {{\mathbf{q}}_D}} \right\|}^2} + {{z_R^2}}}}, \label{hrd}\\
			h_{ED}^2 &= \frac{{{\beta _0}}}{{{{\left\| {{{\mathbf{q}}_E}[n] - {{\mathbf{q}}_D}} \right\|}^2} + z_E^2}}. \label{hed}
		\end{align}
	\end{subequations}

	It is assumed that the self-interference is perfect eliminated \cite{ZengY2016TSP}-\cite{XuJ2017WC},  
	then the received signal at $R$ is expressed as
	\begin{align}\label{yR}
		y_R=\sqrt {{P_S}} {h_{SR}}{x_S} + {n_R},
	\end{align}
	where $x_S$ denotes the source signal, $P_S$ is the transmit power, 
	and 
	$n_R\left[ n \right]\sim CN(0,\sigma_R^2)$ is the additive white Gaussian noise (AWGN).
	Similar to \cite{ZengY2016TSP, ChenJ2022TVT, GeY2022IOT}, the received signal at $D$ is expressed as 
	\footnote{
		For the systems with FD relay, the self-interference (SI) is assumed to experience Gaussian \cite{Feizi2020TCOM}, \cite{Moon2018TWC}. Based on the method in \cite{HuaM2020TCOM} and \cite{LeiH2023ArXivHaos}, it can be observed that the achievable rate considering SI has a similar expression to that without considering SI. 
		In other words, the results in this work also fit the scenarios wherein the SI is considered. 
		Thus, like \cite{ZhangH2020TWC}, \cite{LuH2019TVT}, and \cite{HuG2020CL}, the SI is ignored to simplify the analysis in this work. 
	}
	\begin{align}\label{yD}
		{y_D} &= \sqrt {{P_S}} K{h_{SR}}{h_{RD}}{x_S} + K{h_{RD}}{n_R} \nonumber\\
		&+ \sqrt {{P_E}} {h_{ED}}{x_{\mathrm{AN}}} + {n_D},
	\end{align}
	where 
	$K = \sqrt {\frac{{{P_R}}}{{{\sigma_R^2} + {P_S}h_{SR}^2}}}$ denotes the amplification coefficient, 
	$P_R$ is the transmit power at $R$, 
	$ {x_{\mathrm{AN}}}$ is the AN from $E$, 
	and 
	$n_D\left[ n \right]\sim CN(0,\sigma_D^2)$ is the AWGN at $D$. 
	Then the SNR of $D$ is expressed as
	\begin{align}\label{snrD}
		{\gamma _D}\left[ n \right]  = \frac{{{P_S}{K^2}h_{SR}^2h_{RD}^2}}{{\left( {1 + {K^2}h_{RD}^2} \right){\sigma ^2} + {P_E}\left[ n \right]h_{ED}^2\left[ n \right]}}.
	\end{align}
	
	Similarly, the received signal at $E$ is given by
	\begin{align}\label{yE}
		{y_E}  &= \sqrt {{P_S}} K{h_{SR}}{h_{RE}}{x_S} \nonumber\\
		&+ \sqrt {{P_S}} {h_{SE}}{x_S} + K{h_{RE}}{n_R} + {n_E},
	\end{align}
	where $n_E \sim  CN(0,\sigma_E^2)$. 
	To facilitate the analysis, it is assumed that $\sigma_R^2 = \sigma_D^2 = \sigma^2$. 
	Then the SNR at $E$ is obtained as
	\begin{align}\label{snrE}
		{\gamma _E}\left[ n \right] = \frac{{{P_S}{{\left| {{h_{SE}}\left[ n \right] + K{h_{SR}}{h_{RE}}\left[ n \right]} \right|}^2}}}{{\left( {1 + {K^2}h_{RE}^2\left[ n \right]} \right){\sigma ^2}}}.
	\end{align}
	
	The achievable rates at the receivers are expressed as
	\begin{subequations}
		\begin{align}
			{R_D}\left[ n \right] &= {\log _2}\left( {1 + {\gamma _D}\left[ n \right]} \right), \label{eq:RD}\\
			{R_E}\left[ n \right] &= {\log _2}\left( {1 + {\gamma _E}\left[ n \right]} \right). \label{eq:RE}
		\end{align}
	\end{subequations}
In the PE scenarios, the relationship between the achievable rate between the main and the eavesdropping links is crucial. When $R_E \ge R_D$, surveillance is successful, and $R_D$ is called ER. Otherwise, surveillance fails, and the ER is equal to zero. 	
To ensure successful eavesdropping, the condition ${R_E} \ge {R_D}$ must be satisfied, and $E$ is assumed to decode eavesdropping information with an arbitrarily small error in this work \cite{ZengY2016TSP}, \cite{XuD2022TWC}, \cite{Feizi2020TCOM}, \cite{LuH2019TVT}, \cite{JiangX2017SPL}
	
{ The horizontal flight power of the UAV is expressed as \cite{ZengY2019TWC}
\begin{align}\label{power_hor}
	{P_{\mathrm{hor}}}(\left\| {{{\mathbf{v}}_{xy}}} \right\|) =& {P_0}\left( {1 + \frac{{3{{\left\| {{{\mathbf{v}}_{xy}}} \right\|}^2}}}{{U_{\mathrm{tip}}^2}}} \right) +\frac{1}{2}{d_0}\rho sA{\left\| {{{\mathbf{v}}_{xy}}} \right\|^3} \nonumber\\+&	 {P_i}{\left( {\sqrt {1 + \frac{{{{\left\| {{{\mathbf{v}}_{xy}}} \right\|}^4}}}{{4v_0^4}}}  - \frac{{{{\left\| {{{\mathbf{v}}_{xy}}} \right\|}^2}}}{{2v_0^2}}} \right)^{1/2}},
\end{align}
where 
$\mathbf{v}_{xy}$ denote the horizontal velocity, $P_0$ and $P_i$ denote blade profile power and induced power, respectively, 
$U_{\mathrm{tip}}$ is the tip speed of the rotor blade,
$v_0$ is represented as the mean rotor induced velocity in hover, 
the fuselage drag ratio and rotor solidity are denoted as $d_0$ and $s$, respectively, 
and $\rho$ and $A$ represent the air density and rotor disc area, respectively.
The vertical flight power of the UAV is expressed as \cite{ FilipponeABook}
\begin{align}\label{power_ver}
		{P_{\mathrm{ver}}}({v_z}) = W{v_z},{v_z} > 0,
\end{align}
where 
$W$ denotes the weight of the UAV , $v_z$ denotes the vertical speed. When the vertical speed $v_z$ of the UAV is less than zero, $P_{\mathrm{ver}}=0$.
}

\section{Problem Formulation}
\label{sec:Problem1}
	
This work maximizes the AER of $E$ by jointly optimizing $E$'s horizontal trajectory ${\mathbf{q}}_E\left[ n \right]$ and transmission power ${P_E}\left[ n \right]$ of AN, while the condition for successful eavesdropping is ensured.
Let ${\mathbf{Q}} = \left\{ {{{\mathbf{q}}_E}\left[ n \right], \forall n} \right\}$ and  ${\mathbf{P}} = \left\{ {{P_E}\left[ n \right],\forall n} \right\}$, the AER maximization problem is formulated as
\begin{subequations}
	\begin{align}
		\mathcal{P}_{1} \,:\, &\max_{\mathbf{Q}, \mathbf{P}, \mathbf{v}_{xy}, \mathbf{a}_{xy}}  \Theta \\
		{\mathrm{s.t.}}\;&{R_D} \left[ n \right]\le {R_E}\left[ n \right],  	\label{p1b}\\
		& \mathbf{q}_E\left[ 0 \right] = {{\mathbf{q}}_{\mathrm{I}}}, \mathbf{q}_E\left[ N \right] = {{\mathbf{q}}_{\mathrm{F}}}, \label{p1c}\\
		& \mathbf{q}_E\left[ {n + 1} \right] = {{\mathbf{q}}_E}\left[ n \right] + \mathbf{v}_{xy}\left[ n \right]{\delta _t} + \frac{1}{2}\mathbf{a}_{xy}\left[ n \right]\delta _t^2, \label{p1d}\\
		&\mathbf{v}_{xy}\left[ 0 \right] =\mathbf{v}_{xy}\left[ N \right] \label{p1e}\\
		& \mathbf{v}_{xy}\left[ {n + 1} \right] = \mathbf{v}_{xy}\left[ n \right] + \mathbf{a}_{xy}\left[ n \right]{\delta _t}, \label{p1f}\\
		& \left\| \mathbf{v}_{xy}\left[ {n} \right] \right\|\leq v^{\max}_{xy},\left\| \mathbf{a}_{xy}\left[ {n} \right] \right\|\leq a^{\max}_{xy}, \label{p1g}\\
		& {\frac{1}{N}\sum\limits_{n = 1}^N {{P_{\mathrm{hor}}}(\left\| {{{\mathbf{v}}_{xy}}\left[ n \right]} \right\|} ) \le P_{\mathrm{hor}}^{\mathrm{ave}}, }\label{p1h}	\\
		& 0 \le {P_E}\left[ n \right] \le P_E^{\max }, \label{p1i}		
	\end{align}
\end{subequations}
where 
$\Theta  = \frac{1}{N}\sum\limits_{n = 0}^N {{R_D}\left[ n \right]} $ denotes the AER, 
${\mathbf{q}}_I$ and ${\mathbf{q}}_F$ signify the initial and final positions of $E$, respectively,
$\mathbf{v}_{xy}\left[ n \right]$ and $\mathbf{a}_{xy}\left[ n \right]$ denote the horizontal velocity and acceleration, respectively, 
$v^{\max}_{xy}$ denote the horizontal maximum velocity,
{{
		$P_{\mathrm{hor}}^{\mathrm{ave}}$ represents the average horizontal flight power of $E$.}
	(\ref{p1b}) represents the constraint of successful eavesdropping, 
	(\ref{p1c}) represents the constraint of horizontal location of the initial point and the final point of $E$, respectively, 
	(\ref{p1d})-(\ref{p1g}) are the constraints of the horizontal position, horizontal velocity, and horizontal acceleration of $E$, 
	(\ref{p1h}) represents the constraint of horizontal flight power, 
	(\ref{p1i}) represents the constraint of maximum jamming power, 
}
and $P_E^{\max }$ is the maximum transmit power of $E$.
	
It can be observed that $\mathcal{P}_{1}$ is difficult to obtain the solution directly due to the non-convex constraints (\ref{p1b}), (\ref{p1h}), and the strong coupling of $\mathbf {Q}$ and $\mathbf {P}$. 
Firstly, the left-hand side (LHS) of (\ref{p1b}), ${R_D}\left[ n \right]$ is a function of $\mathbf{Q}$ and $\mathbf{P}$ respectively. The concave and convex properties of the function of $\mathbf{Q}$ and $\mathbf{P}$ are opposite. Thus, it is a challenge to determine whether the LHS of (\ref{p1b}) satisfies the requirement of a convex function. 
Secondly, the right-hand side (RHS) of (\ref{p1b}), ${R_E}\left[ n \right]$, is a function about $\mathbf{Q}$ and is too complex to determine the concavity and convexity.  
{{Thirdly, the LHS of (\ref{p1h}) is complex and does not satisfy the convex constraint.  }}
Therefore, solving the original problem $\mathcal{P}_{1}$ is highly non-trivial.	
	
\section{Proposed Algorithm for Problem $\mathcal{P}_{1}$}
\label{sec:ProposedAlgorithm1}

To solve $\mathcal{P}_{1}$, alternating optimization method is utilized to optimize $\mathbf {Q}$ and $\mathbf {P}$ in an alternating manner, by considering the others to be given.
Specifically, with given the jamming power $\mathbf {P}$ of $E$, the horizontal flying trajectory $\mathbf {Q}$ of $E$ is optimized. 
For any given $E$'s trajectory $\mathbf {Q}$, the jamming power $\mathbf {P}$ of $E$ is optimized.

\subsection{Subproblem 1: Optimizing horizontal position}
\label{Subproblem01}

In this subsection, ${\mathbf{Q}}$ is optimized with given fixed variables ${{\mathbf{P}}}$. The problem $\mathcal{P}_{1}$ is reduced as
\begin{subequations}
	\begin{align}
		\mathcal{P}_{1.1} \,:\ & \mathop {\max }\limits_{\mathbf{Q},\mathbf{v}_{xy}, \mathbf{a}_{xy}} \Theta\\
		{\mathrm{s.t.}}\; & (\textrm{\ref{p1b}})-(\textrm{\ref{p1h}}). \notag 
	\end{align}
\end{subequations}

$\mathcal{P}_{1.1}$ is not a standard convex problem and cannot be solved by the convex optimization toolbox since 
${R_D}\left[ n \right]$ 
is not a convex function of ${\mathbf{Q}}$.
To deal with this problem, (\ref{p1b}) is removed in this subsection and solved by designing the jamming power in the following subsection.
Firstly, $R_E\left[ n \right]$ is rewritten as
\begin{align}\label{RE2}
	{R_E}\left[ n \right] = {R_{E1}}\left[ n \right] - {R_{E2}}\left[ n \right],
\end{align}
where
${R_{E1}}\left[ n \right] = {\log _2}\left(1 + \frac{{{\beta _0}{K^2}}}{{d_{RE}\left[ n \right]}} + {g_1}\left[ n \right]{\rho_S}\right)$,
${R_{E2}}\left[ n \right] $ $= {\log _2}\left(1 + \frac{{{\beta _0}{K^2}}}{{d_{RE}\left[ n \right]}}\right)$,
${g_1}\left[ n \right] = \sqrt {\frac{{4\beta _0^3{K^2}}}{{{d_{SE}\left[ n \right]}{d_{SR}}{d_{RE}\left[ n \right]}}}}+\frac{{\beta _0^2{K^2}}}{{{d_{SR}}{d_{RE}\left[ n \right]}}} +$ $\frac{{{\beta _0}}}{{d_{SE}\left[ n \right]}}$,
${d_{SE}}\left[ n \right]  = \left\| {{{\mathbf{q}}_E}\left[ n \right] - {{{\mathbf{q}}_S}}} \right\|^2+ z_E^2$,
${d_{ED}}\left[ n \right] = \left\| {{{\mathbf{q}}_E}\left[ n \right] - {{{\mathbf{q}}_D}}} \right\|^2+ $ $ z_E^2$,
${d_{RE}}\left[ n \right] = \left\| {{{\mathbf{q}}_E}\left[ n \right] - {{{\mathbf{q}}_R}}} \right\|^2+ {\left\|{{z_E}}-{z_R}\right\|}^2$, 
${d_{SR}}  = \left\| {{{\mathbf{q}}_R} - {{{\mathbf{q}}_S}}} \right\|^2+ z_R^2$, 
and
${\rho_S}=\frac{P_S}{\sigma^2}$.

{{
To transform the LHS of (\ref{p1h}) into a convex limit, we introduce the relaxation variable ${\upsilon _n} \ge \left\| {{{\mathbf{v}}_{xy}}\left[ n \right]} \right\|$.
By substituting the $\upsilon _n$ into (\ref{power_hor}), we obtain 
\begin{align}\label{power_hor1}
	{P_{\mathrm{hor}}}({\upsilon _n}) =& {P_0}\left( {1 + \frac{{3{\upsilon _n}^2}}{{U_{\mathrm{tip}}^2}}} \right) + {P_i}{\left( {\sqrt {1 + \frac{{{\upsilon _n}^4}}{{4v_0^4}}}  - \frac{{{\upsilon _n}^2}}{{2v_0^2}}} \right)^{1/2}} \nonumber\\&+ \frac{1}{2}{d_0}\rho sA{\upsilon _n}^3,
\end{align}
where the second term in (\ref{power_hor1}) is also a non-convex constraint. By introducing the relaxation variable $\tau_n$, we have 
\begin{align}
	\tau _n^2 &\ge \sqrt {1 + \frac{{{\upsilon _n}^4}}{{4v_0^4}}}  - \frac{{{\upsilon _n}^2}}{{2v_0^2}} \nonumber\\
	 \Leftrightarrow \frac{1}{{\tau _n^2}} &\le \tau _n^2 + \frac{{{\upsilon _n}^4}}{{v_0^2}}.
	\label{power_tau}
\end{align}			
The lower bound of the RHS of (\ref{power_tau}) is expressed as
 \begin{align}
 	\tau _n^2 + \frac{{{\upsilon _n}^2}}{{v_0^2}} &\ge 2\tau _n^{\left( j \right)}\left( {{\tau _n} - \tau _n^{\left( j \right)}} \right) + \frac{{{{\left( {\upsilon _n^{\left( j \right)}} \right)}^2}}}{{v_0^2}} \nonumber\\
 	&+ 2\frac{{\upsilon _n^{\left( j \right)}}}{{v_0^2}}\left( {{\upsilon _n} - \upsilon _n^{\left( j \right)}} \right) = \varphi _{\tau_n} ^{\mathrm{lb}},
\end{align}
where ${\left(  \right)} ^{\left( j \right)}$ denotes $j$-th iteration value.
Based on the above transformation, (\ref{p1h}) is rewritten as
\begin{align}
  	\frac{1}{N}\sum\limits_{n = 1}^N &{\left( {{P_0}\left( {1 + \frac{{3{\upsilon _n}^2}}{{U_{\mathrm{tip}}^2}}} \right) + {P_i}\tau _n  + \frac{1}{2}{d_0}\rho sA{\upsilon _n^3}} \right)}  \le P_{\mathrm{hor}}^{\mathrm{ave}},  \label{power_hor2}\\
  	&\frac{1}{{\tau _n^2}} \le \varphi _{\tau_n} ^{\mathrm{lb}}, \label{power_hor3}\\
  	&{\upsilon _n} \ge \left\| {{{\mathbf{v}}_{xy}}\left[ n \right]} \right\|. \label{power_hor4} 
\end{align}
}
}	  	
	Then, by introducing slack variables $\mathbf{S_1}=\left\{S_{1}\left[ n \right],S_{2}\left[ n \right]\},S_{3}\left[ n \right], \forall n \right\}$, 
	$\mathcal{P}_{1.1}$ is rewritten as 
	\begin{subequations}
		\begin{align}
			\mathcal{P}_{1.{\mathrm{1a}}} \,:\  	&\mathop {\max }\limits_{\mathbf{Q},\mathbf{S_1},\mathbf{v}_{xy}, \mathbf{a}_{xy}} {\Theta _1} \\
			{\mathrm{s.t.}}\; &S_{1}\left[ n \right] \le {R_D}\left[ n \right], \label{p12b}\\
			& {S_2}\left[ n \right] \le {R_{E1}}\left[ n \right], \label{p12c}\\
			&S_{3}\left[ n \right] \ge {R_{E2}}\left[ n \right], \label{p12d}\\
			&S_{1}\left[ n \right] \le S_{2}\left[ n \right] - S_{3}\left[ n \right], \label{p12e}\\
			&(\textrm{\ref{p1c}})-(\textrm{\ref{p1g}}), (\textrm{\ref{power_hor2}})-(\textrm{\ref{power_hor4}})\notag 
		\end{align}
	\end{subequations}
	where 
	${\Theta _1} = \frac{1}{N}\sum\limits_{n = 0}^N {{S_1}} \left[ n \right]$. 
	However, the expressions of (\ref{p12b}), (\ref{p12c}) and (\ref{p12d})  are still nonconvex constraints on $\mathbf{q}_E\left[ n \right]$. 
	By introducing slack variables $\mathbf{d}=\left\{{\hat d_{SE}}\left[ n \right],{\hat d_{RE}}\left[ n \right], {\hat d_{ED}}\left[ n \right], \forall n \right\}$, $\mathcal{P}_{1.{\mathrm{1a}}}$ is rewritten as
	\begin{subequations}
		\begin{align}
			\mathcal{P}_{1.{\mathrm{1b}}} \,:\  	&\mathop {\max }\limits_{\mathbf{Q},\mathbf{S_1},\mathbf{d},\mathbf{v}_{xy}, \mathbf{a}_{xy}} {\Theta _1} \\
			{\mathrm{s.t.}}\; &S_{1}\left[ n \right] \le {\hat R_D}\left[ n \right], \label{p13b}\\
			& {S_2}\left[ n \right] \le {{\hat R}_{E1}}\left[ n \right], \label{p13c}\\
			&S_{3}\left[ n \right] \ge {{\hat R}_{E2}}\left[ n \right], \label{p13d}\\
			&S_{1}\left[ n \right] \le S_{2}\left[ n \right] - S_{3}\left[ n \right], \label{p13e}\\
			&{\hat d_{SE}}\left[ n \right] \ge {d_{SE}}\left[ n \right], \label{p13f}\\
			&{\hat d_{RE}}\left[ n \right] \ge {d_{RE}}\left[ n \right], \label{p13g}\\
			&{\hat d_{ED}}\left[ n \right] \le {d_{ED}}\left[ n \right], \label{p13h}\\
			& (\textrm{\ref{p1c}}) - (\textrm{\ref{p1g}}), (\textrm{\ref{power_hor2}})-(\textrm{\ref{power_hor4}}), \notag 
		\end{align}
	\end{subequations}
	where 
	${\hat R_D}\left[ n \right] = {\log _2}\left(1 + \frac{{{K^2}h_{SR}^2 h_{RD}^2{P_S}}}{{(1+{K^2}h_{RD}^2){\sigma ^2} + \frac{{{P_E}\left[ n \right]{\beta _0}}}{{{{\hat d_{ED}}}\left[ n \right]}}}}\right)$	
	${{\hat R}_{E1}}\left[ n \right] = {\log _2}\left( {1 + {\rho _S}\sqrt {\frac{{4\beta _0^3{K^2}}}{{{{\hat d}_{SE}}\left[ n \right]{d_{SR}}{{\hat d}_{RE}}\left[ n \right]}}}  + \frac{{{\beta _0}}}{{{{\hat d}_{SE}}\left[ n \right]}}} \right.$ $\left. { + \frac{{\beta _0^2{K^2}}}{{{d_{SR}}{{\hat d}_{RE}}\left[ n \right]}} + \frac{{{\beta _0}{K^2}}}{{{{\hat d}_{RE}}\left[ n \right]}}} \right)$, 
	and 
	${{\hat R}_{E2}}\left[ n \right] =  {\log _2}\left(1 + \frac{{{\beta _0}{K^2}}}{{{{\hat d_{RE}}}\left[ n \right]}}\right)$.

	To determine the concavity and convexity of the RHS of (\ref{p13c}), 
	\textit{Lemma 1} is given as following.
	
	\begin{lemma}
		For given $c_1 > 0, c_2 > 0,c_3 > 0,c_4 > 0$, $f = {\log _2}\left({c_1} + {c_2}{x^{ - 1}} + {c_3}{y^{ - 1}} + {c_4}{x^{ - \frac{1}{2}}}{y^{ - \frac{1}{2}}}\right)$ is a convex function.
	\end{lemma}

	\begin{proof}
		See Appendix \ref{sec:appendicesA}.
	\end{proof}
	
	Based on \textit{Lemma 1}, ${{\hat R}_{E1}}\left[ n \right]$ is a convex function of ${\hat d_{SE}}$ and ${\hat d_{RE}}$, which does not satisfy the convex constraint. By utilizing the SCA technology, the lower bound of ${{\hat R}_{E1}}\left[ n \right]$ can be recast as 
	\begin{align}\label{hatRE1}
		&{{\hat R}_{E1}}\left[ n \right] \ge {\log _2}{A_1}\left[ n \right] + \frac{{{B_1}\left[ n \right]}}{{\ln 2{A_1}\left[ n \right]}}\left({\hat d_{RE}}\left[ n \right]-\hat d_{RE}^{\left( j \right)}\left[ n \right]\right) \nonumber\\
		&\,\, + \frac{{{C_1}\left[ n \right]}}{{\ln 2{A_1}\left[ n \right]}}\left({\hat d_{SE}}\left[ n \right]-\hat d^{\left( j \right)}_{SE}\left[ n \right]\right) \nonumber\\
		&\,\, \buildrel \Delta \over = \varphi _{RE1}^{\mathrm{lb}}\left[ n \right],
	\end{align}
	where 
	${A_1}\left[ n \right] = 1 + \frac{{{\beta _0}{K^2}}}{{\hat d_{RE}^{\left( j \right)}\left[ n \right]}} + \frac{{{\rho _S}{\beta _0}}}{{\hat d_{SE}^{\left( j \right)}\left[ n \right]}} + \frac{{{\rho _S}\beta _0^2{K^2}}}{{{d_{SR}}\hat d_{RE}^{\left( j \right)}\left[ n \right]}} + {\rho _S}\sqrt {\frac{{4\beta _0^3{K^2}}}{{{d_{SR}}\hat d_{SE}^{\left( j \right)}\left[ n \right]\hat d_{RE}^{\left( j \right)}\left[ n \right]}}} $, 
	${B_1}\left[ n \right] =  - \frac{{{\beta _0}{K^2}}}{{{{\left( {\hat d_{RE}^{\left( j \right)}\left[ n \right]} \right)}^2}}} - \frac{{\beta _0^2{K^2}{\rho _S}}}{{{d_{SR}}{{\left( {\hat d_{RE}^{\left( j \right)}\left[ n \right]} \right)}^2}}} -$ $ {\rho _S}\sqrt {\frac{{4\beta _0^3{K^2}}}{{\hat d_{SE}^{\left( j \right)}\left[ n \right]{d_{SR}}{{\left( {\hat d_{RE}^{\left( j \right)}\left[ n \right]} \right)}^3}}}} $, 
	${C_1}\left[ n \right] =  - \frac{{{\beta _0}{\rho _S}}}{{{{\left( {\hat d_{SE}^{\left( j \right)}\left[ n \right]} \right)}^2}}} - {\rho _S}\sqrt {\frac{{4\beta _0^3{K^2}}}{{{d_{SR}}\hat d_{RE}^{\left( j \right)}\left[ n \right]{{\left( {\hat d_{SE}^{\left( j \right)}\left[ n \right]} \right)}^2}}}}  $, 
	${\hat d_{RE}}^{\left( j \right)}$ and ${\hat d_{SE}}^{\left( j \right)}$ represent the values of the $j$th iteration of  ${\hat d_{RE}}$ and ${\hat d_{SE}}$, respectively. 
	
	Similarly, the lower bound of the RHS of (\ref{p13h}) is obtained as
	\begin{align}\label{RHSof13}
		&{\left\| {{{\mathbf{q}}_E}\left[ n \right] - {{\mathbf{q}}_D}} \right\|^2} +  z_E^2 \ge {\left\| {\mathbf{q}^{\left( j \right)}\left[ n \right] - {{\mathbf{q}}_D}} \right\|^2} + z_E^2 \nonumber\\
		&\;\;\;\;\;\;  + 2\left\| {{\mathbf{q}^{\left( j \right)}}\left[ n \right] - {{\mathbf{q}}_D}} \right\|\left({{\mathbf{q}}_E}\left[ n \right] - {\mathbf{q}^{\left( j \right)}}\left[ n \right]\right) \nonumber\\
		&\;\;\;\;\;\;  \buildrel \Delta \over = {\varphi ^\mathrm{lb}_{ued}}\left[ n \right],
	\end{align}
	where $\mathbf{q}^{\left( j \right)}\left[ n \right]$ denotes the value of the $j$th iteration of ${{\mathbf{q}}_E}\left[ n \right]$. 
	
	To facilitate the solution, ${\hat R_D}\left[ n \right]$ and ${{\hat R}_{E2}}\left[ n \right]$ are also relaxed as 
	\begin{align}\label{hatRD}
		{\hat R_D}\left[ n \right] &\le {\log _2}{A_0}\left[ n \right] + \frac{{{B_0}\left[ n \right]}}{{\ln 2{A_0}\left[ n \right]}}\left({{\hat d_{ED}}}\left[ n \right] - {\hat d_{ED}}^{\left( j \right)}\left[ n \right]\right) \nonumber\\
		&\buildrel \Delta \over =  \varphi _{RD}^{\mathrm{ub}}\left[ n \right]
	\end{align}
	and
	\begin{align}\label{hatRE2}
		{{\hat R}_{E2}}\left[ n \right] &\ge \frac{{{B_2}\left[ n \right]}}{{\ln 2{A_2}\left[ n \right]}}\left({\hat d_{RE}}\left[ n \right]-\hat d_{RE}^{\left( j \right)}\left[ n \right]\right)  +  {\log _2}{A_2}\left[ n \right] \nonumber\\
		&\buildrel \Delta \over = \varphi _{RE2}^\mathrm{lb}\left[ n \right],
	\end{align}
	where 
	${A_0}\left[ n \right] = 1 + \frac{{{K^2}h_{SR}^2h_{RD}^2{P_S}}}{{\left(1+{K^2}h_{RD}^2\right){\sigma ^2} + \frac{{{P_E}\left[ n \right]{\beta _0}}}{{{\hat d_{ED}}^{\left( j \right)}\left[ n \right]}}}}$,
	${B_0}\left[ n \right] = \frac{{{K^2}h_{SR}^2h_{RD}^2{P_S}\frac{{{P_E}\left[ n \right]{\beta _0}}}{{{{({\hat d_{ED}}^{\left( j \right)}\left[ n \right])}^2}}}}}{{{{\left(\left(1 + {K^2}h_{RD}^2\right){\sigma ^2} + \frac{{{P_E}\left[ n \right]{\beta _0}}}{{{\hat d_{ED}}^{\left( j \right)}\left[ n \right]}}\right)}^2}}}$,
	${A_2}\left[ n \right] =1+ \frac{{{\beta _0}{K^2}}}{{{\hat d_{RE}^{\left( j \right)}}\left[ n \right]}}$,
	${B_2}\left[ n \right] =  - \frac{{{\beta _0}{K^2}}}{{{{\left({\hat d_{RE}^{\left( j \right)}}\left[ n \right]\right)}^2}}}$,
	${\hat d_{ED}}^{\left( j \right)}$ denotes the value of the $j$th iteration of ${\hat d_{ED}}$, 
	and 
	the superscript `ub' signifies the upper bound.
	
	Then, $\mathcal{P}_{1.{\mathrm{1b}}}$ is reformulated as
	\begin{subequations}
		\begin{align}
			\mathcal{P}_{1.{\mathrm{1c}}} \,:\  	& \mathop {\max }\limits_{\mathbf{Q},\mathbf{S_1},\mathbf{d},\mathbf{v}_{xy}, \mathbf{a}_{xy}} {\Theta _1}  \label{p11ca}\\
			{\mathrm{s.t.}}\; & S_{1}\left[ n \right] \le \varphi _{RD}^\mathrm{ub}\left[ n \right], \label{p11cb}\\
			&S_{2}\left[ n \right] \le \varphi _{RE1}^\mathrm{lb}\left[ n \right],\label{p11cc}\\
			&S_{3}\left[ n \right] \ge \varphi _{RE2}^\mathrm{lb}\left[ n \right],\label{p11cd}\\
			&S_{1}\left[ n \right] \le S_{2}\left[ n \right] - S_{3}\left[ n \right], \label{p11ce}\\
			&{\hat d_{SE}}\left[ n \right] \ge {d_{SE}}\left[ n \right], \label{p11cf}\\
			&{\hat d_{RE}}\left[ n \right] \ge {d_{RE}}\left[ n \right], \label{p11cg}\\
			&{{\hat d_{ED}}}\left[ n \right] \le {\varphi ^\mathrm{lb}_{ued}}\left[ n \right],\label{p11ch}\\
			& (\textrm{\ref{p1c}})- (\textrm{\ref{p1g}}), (\textrm{\ref{power_hor2}})-(\textrm{\ref{power_hor4}}). \notag 
		\end{align}
	\end{subequations}
	$\mathcal{P}_{1.{\mathrm{1c}}}$ is a standard convex problem that can be efficiently solved by existing standard convex optimization tools such as CVX.

	\subsection{Subproblem 2: Optimizing Jamming Power} 
	\label{Subproblem02}

	With the given $\mathbf{Q}$, subproblem $\mathcal{P}_{1.2}$ is formulated as 
	\begin{small}
		\begin{subequations}
			\begin{align}
				\mathcal{P}_{1.2} \,:\	&\mathop {\max }\limits_{\mathbf{P}} \Theta  \label{P20a}\\
				{\mathrm{s.t.}}\; &{R_D}\left[ n \right] \le {R_E}\left[n\right], \label{P20b}\\
				& (\textrm{\ref{p1i}}). \notag 
			\end{align}
		\end{subequations}
	\end{small}
	
	$\mathcal{P}_{1.2}$ is not a standard convex problem because $R_D$ is a convex function of $P_E$ in (\ref{P20a}). 
	The lower bound of $R_D$ is denoted as
	\begin{align}\label{LBRD}
		R_D\left[ n \right]	&\ge {\log _2}A_2\left[ n \right] + \frac{{B_2\left[ n \right]}}{{\ln 2A_2\left[ n \right]}}({P_E}\left[ n \right] - P_E^{\left( j \right)}\left[ n \right]) \nonumber\\
		&\buildrel \Delta \over = \varphi _D^\mathrm{lb}\left[ n \right],
	\end{align}
	where  $A_2\left[ n \right] = 1 + \frac{{{K^2}h_{SR}^2 h_{RD}^2 {P_S}}}{{\left(1+{K^2}h_{RD}^2\right){\sigma ^2} + {P_E^{\left( j \right)}}\left[ n \right]h_{ED}^2\left[ n \right]}}$, $B_2\left[ n \right] =  - \frac{{{K^2}h_{SR}^2 h_{RD}^2 h_{ED}^2\left[ n \right]{P_S}}}{{{{\left( {(1+{K^2}h_{RD}^2){\sigma ^2} + {P_E^{\left( j \right)}}\left[ n \right]h_{ED}^2\left[ n \right]} \right)}^2}}}$, and 
	$P_E^{\left( j \right)}\left[ n \right]$ denotes the $j$th iteration of $P_E\left[ n \right]$.
	
	Although (\ref{P20b}) satisfies the convex constraint, to facilitate the operation of the solution in the CVX toolbox, $R_{D}\left[ n \right]$ is expressed as
	\begin{align}\label{RD02}
		R_D\left[ n \right] = R_{D1}\left[ n \right] - R_{D2}\left[ n \right],
	\end{align}
	where 
	$R_{D1}\left[ n \right] = {\log _2}\left((1+{K^2}h_{RD}^2){\sigma ^2} + {P_E}\left[ n \right]h_{ED}^2\left[ n \right] +\right.$ $\left.{K^2}h_{SR}^2h_{RD}^2{P_S} \right)$ 
	and 
	$R_{D2}\left[ n \right] = {\log _2}\left((1+{K^2}h_{RD}^2){\sigma ^2} + \right.$ $\left.{P_E}\left[ n \right]h_{ED}^2\left[ n \right]\right)$.  
	$R_{D1}\left[ n \right]$ is concave functions of $P_E$ and its upper bound is obtained as
	\begin{align}\label{UBRD1}
		R_{D1}\left[ n \right] &\le	 {\log_2}A_3\left[ n \right] + \frac{{B_3\left[ n \right]}}{{\ln 2A_3\left[ n \right]}}({P_E}\left[ n \right] - P_E^{\left( j \right)}\left[ n \right]) \nonumber\\
		&= \varphi _{D1}^\mathrm{ub}\left[ n \right],
	\end{align}
	where 
	$A_3\left[ n \right]=(1+{K^2}h_{RD}^2){\sigma ^2} + {P_E^{\left( j \right)}}\left[ n \right]h_{ED}^2\left[ n \right] +{K^2}h_{SR}^2h_{RD}^2{P_S}$
	and 
	$B_3\left[ n \right] = h_{ED}^2\left[ n \right]$.

	$\mathcal{P}_{1.{\mathrm{2}}}$ can be rewritten as
	\begin{subequations}
		\begin{align}
			\mathcal{P}_{1.{\mathrm{2a}}} \,:\	& \mathop {\max }\limits_{\mathbf{P}}  {\Theta ^{{\mathrm{lb}}}}\\
			{\mathrm{s.t.}}\;& \varphi _{D1}^\mathrm{ub}\left[ n \right] - R_{D2}\left[ n \right] \le {R_E}\left[n\right], \label{p12ab}\\
			& (\textrm{\ref{p1i}}), \notag 
		\end{align}
	\end{subequations}
	where 
	${\Theta ^{{\mathrm{lb}}}} = \frac{1}{N}\sum\limits_{n = 0}^N {\varphi _D^{{\mathrm{lb}}}\left[ n \right]} $. 
	$\mathcal{P}_{1.{\mathrm{2a}}}$ is a standard convex problem and can be solved with existing optimization tools.
	
	\textbf{Algorithm 1} summarizes the details of overall iterations for $\mathcal{P}_{1}$. 
	
	\begin{algorithm}[t]
		\caption{An iterative algorithm for joint optimization of 2D trajectory and jamming power}
		\label{algorithm1}
		\begin{algorithmic}[1]  
			\STATE \textbf{Input:} ${{\mathbf{q}}_S}$, ${{\mathbf{q}}_D}$, $\mathbf{q}_I$, $\mathbf{q}_F$, $a^{\max}_{xy}$, $v^{\max}_{xy}$, $H$, $T$
			\STATE \textbf{Output:} $\mathbf{Q}$, $\mathbf{P}$
			\STATE  Initialization: set initial variables; tolerance $\varepsilon > 0$ and iteration number $j=0$.
			\STATE \textbf{Repeat}    
			\STATE With given $\mathbf{P}^{{\left( {j} \right)}}$, obtain $\mathbf{Q}^{{\left( {j + 1} \right)}}$, $\mathbf{S_1}^{{\left( {j + 1} \right)}}$, $\mathbf{d}^{{\left( {j + 1} \right)}}$ by solving problem ($\mathcal{P}_{1.{\mathrm{1b}}}$)
			\STATE With given $\mathbf{Q}^{{\left( {j + 1} \right)}}$, $\mathbf{S_1}^{{\left( {j + 1} \right)}}$, $\mathbf{d}^{{\left( {j + 1} \right)}}$, obtain $\mathbf{P}^{{\left( {j + 1} \right)}}$ by solving problem ($\mathcal{P}_{1.{\mathrm{2a}}}$)
			\STATE $j=j+1$
			\STATE \textbf{Until} {$\left| {{\Theta ^{\left( {j + 1} \right)}} - {\Theta ^{\left( j \right)}}} \right| \le \varepsilon $ or reaches the maximum number of iterations}
			\STATE Obtain the optimal solution
		\end{algorithmic} 
	\end{algorithm}
	
	\section{Joint Optimization 3D Trajectory and Jamming Power of $E$} 
	\label{sec:Problem2}
	
	In this section, it is considered that $E$ flies at a varying altitude $z_E$ and the 3D trajectory and jamming power of $E$ are jointly optimized.
	Let ${\mathbf{Z}} = \left\{ {{{z_E}\left[ n \right]},\forall n} \right\}$, the optimization problem is formulated as
	\begin{subequations}
		\begin{align}
			\mathcal{P}_{2} \,:\, &\max_{\mathbf{Q}, \mathbf{P},{\mathbf{Z},\mathbf{v}_{xy}, \mathbf{a}_{xy},{v}_{z}, {a}_{z}}} \Theta  \\
		{\mathrm{s.t.}}\; & {{z_E}\left[ 0 \right]} = {H_{\mathrm{I}}}, {{z_E}\left[ N \right]} = {H_{\mathrm{F}}}, \label{p2b}\\
		& {z_E}\left[ {n + 1} \right] = {{z_E}\left[ n \right]} + v_{z}\left[ n \right]{\delta _t} + \frac{1}{2}a_{z}\left[ n \right]\delta _t^2, \label{p2c}\\
		& {{z_E}\left[ n \right]} \leq H_{\max},{{z_E}\left[ n \right]} \geq H_{\min}, \label{p2d} \\
		& {v}_{z}\left[ {n + 1} \right] = {v}_{z}\left[ n \right] + {a}_{z}\left[ n \right]{\delta _t}, n=0,1,...,N-1, \label{p2e}\\
		& {{v_z}\left[ 0 \right]}  = {{v_z}\left[ N \right]} \label{p2f}\\
		& \left| {{v_z}\left[ n \right]} \right| \le v_z^{\max },\left| {{a_z}\left[ n \right]} \right| \le a_z^{\max }, \label{p2g} \\
	    & {\frac{1}{N}\sum\limits_{n = 1}^N {{P_{\mathrm{ver}}}(\left| {{{{v}}_{z}}\left[ n \right]} \right|} ) \le P_{\mathrm{ver}}^{\mathrm{ave}},\label{p2h}}	\\
		& (\textrm{\ref{p1b}})-(\textrm{\ref{p1i}}), \notag 
		\end{align}
	\end{subequations}
where 
$H_{\max}$ and $H_{\min}$ are the vertical maximum and minimum position of $E$, respectively, 
$v_z\left[ n \right]$ and $a_z\left[ n \right]$ signify the vertical velocity and acceleration, respectively, 
$v^{\max}_{z}$ and $a^{\max}_{z}$ denote the vertical maximum velocity and acceleration of $E$, respectively, 
{{
$P_{\mathrm{ver}}^{\mathrm{ave}}$  represents the average vertical flight power.
	(\ref{p2b}) represents the constraint of vertical location of the initial point and the final point of $E$, respectively, 
	(\ref{p2c})-(\ref{p2g}) are the constraints of the vertical position, vertical velocity, and vertical acceleration of $E$, 
	(\ref{p2h}) represents the constraint of vertical flight power. 
	}}	

$\mathcal{P}_{2}$ is also difficult to solved due to the non-convex constraints (\ref{p1b}) which is coupled by the optimization variable $\mathbf {Q}$, $\mathbf {P}$ and $\mathbf {Z}$.  
To solve $\mathcal{P}_{2}$, the horizontal position of $E$ (${{\mathbf{Q }}}$), vertical trajectory of $E$ (${{\mathbf{H }}}$), and jamming power (${{\mathbf{P}}}$) are iteratively optimized alternately to obtain the maximum AER. 
One can find, optimizing ${{\mathbf{Q }}}$ and ${{\mathbf{P}}}$ can be solved in $\mathcal{P}_{1.{\mathrm{1b}}}$ and in $\mathcal{P}_{1.{\mathrm{2}}}$, respectively.
Thus, the subproblem of optimizing the vertical trajectory of $E$ is given as follows.
	
\subsection{Subproblem 3: Optimizing Vertical Trajectory}
\label{Subproblem3}

Given the horizontal trajectory  $\mathbf{Q}$ and jamming power $\mathbf{P}$, $\mathcal{P}_{2}$ is rewritten as
\begin{subequations}
	\begin{align}
		\mathcal{P}_{2.1} \,:\, &\max_{\mathbf{Q}, \mathbf{P},{\mathbf{Z},{v}_{z}, {a}_{z}}} \Theta _2 \\
		{\mathrm{s.t.}}\; & (\textrm{\ref{p2b}})-(\textrm{\ref{p2h}}), (\textrm{\ref{p1b}}). \notag 
	\end{align}
\end{subequations}
where 
${\Theta _2} = \frac{1}{N}\sum\limits_{n = 0}^N {{S_4}\left[ n \right]} $. 
	
$\mathcal{P}_{2.1}$ cannot be solved directly because (\ref{p1b}) is a non-convex constraint. 
For the convenience of solving, (\ref{p1b}) in $\mathcal{P}_{2.1}$ is also replaced by introducing slack variable  $\mathbf{Z_1}=\left\{z_h\left[ n \right],z_{ER}\left[ n \right], \forall n \right\}$, $\mathbf{S_2}=\left\{S_{4}\left[ n \right],S_{5}\left[ n \right],S_{6}\left[ n \right], \forall n \right\}$, and $\mathcal{P}_{2.1}$ is rewritten as
\begin{subequations}
	\begin{align}	
		\mathcal{P}_{2.2} \,:\  &	\mathop {\max }\limits_{{\mathbf{Z}},\mathbf{S_2},\mathbf{Z_1},{v}_{z}, {a}_{z}} {\Theta _2} \\
		{\mathrm{s.t.}}\; &S_{4}\left[ n \right] \le {{\tilde R}_D}\left[ n \right], \label{p22b}\\
		&S_{5}\left[ n \right] \le {{\tilde R}_{E1}}\left[ n \right], \label{p22c}\\
		&S_{6}\left[ n \right] \ge {{\tilde R}_{E2}}\left[ n \right],\label{p22d}\\
		&S_{4}\left[ n \right] \le S_{5}\left[ n \right] - S_{6}\left[ n \right],\label{p22e}\\
		&{z_h}\left[ n \right] \le z_E^2\left[ n \right],\label{p22f}\\
		&{z_{ER}}\left[ n \right] \ge {\left({{z_E}\left[ n \right]} - {z_R}\right)^2},\label{p22g}\\
		& (\textrm{\ref{p2b}})-(\textrm{\ref{p2h}}), \notag 
	\end{align}
\end{subequations}
where 
${{\tilde R}_D}\left[ n \right] = {\log _2}\left(1 + \frac{{{K^2}h_{SR}^2h_{RD}^2{P_S}}}{{(1+{K^2}h_{RD}^2){\sigma ^2} + \frac{{{P_E}\left[ n \right]{\beta _0}}}{{{{\tilde d}_{ED}}}\left[ n \right]}}}\right)$,
${{\tilde R}_{E1}}\left[ n \right] = {\log _2}\left( {1 + {\rho _S}\sqrt {\frac{{4\beta _0^3{K^2}}}{{{d_{SR}}{{\tilde d}_{SE}}\left[ n \right]{{\tilde d}_{RE}}\left[ n \right]}}}  + \frac{{{\beta _0}{K^2}}}{{{{\tilde d}_{RE}}\left[ n \right]}} + \frac{{{\rho _S}{\beta _0}}}{{{{\tilde d}_{SE}}\left[ n \right]}}} \right.$ $\left. { + \frac{{{\rho _S}\beta _0^2{K^2}}}{{{d_{SR}}{{\tilde d}_{RE}}\left[ n \right]}}} \right)$, 
${{\tilde R}_{E2}}\left[ n \right] = {\log _2}\left(1 + \frac{{{\beta _0}{K^2}}}{{{{\tilde d}_{RE}\left[ n \right]}}}\right)$,
${{{\tilde d}_{ED}}}\left[ n \right]=\left\| {{\mathbf{q}}_E}\left[ n \right] -{{\mathbf{q}}_D}\right\|^2 + {z_h}\left[ n \right]$, 
${{{\tilde d}_{SE}}}\left[ n \right]=\left\| {{\mathbf{q}}_E}\left[ n \right] -{{\mathbf{q}}_S}\right\|^2 + {z_h}\left[ n \right]$, 
and 
${{{\tilde d}_{RE}}}\left[ n \right]=\left\| {{\mathbf{q}}_E}\left[ n \right] -{{\mathbf{q}}_R}\right\|^2 + {z_{ER}}\left[ n \right]$.

	$\mathcal{P}_{2.2}$ cannot be solved by CVX because (\ref{p22c}) and (\ref{p22f}) are not convex.
	Similarly, it can be found that ${{\tilde R}_{E1}}\left[ n \right]$ is a convex function about $z_h$ and $z_{ER}$ based on \textit{Lemma 1}. 
	Therefore, the upper bound of ${{\tilde R}_{E1}}\left[ n \right]$ is obtained as (\ref{UBRE1}), shown at the top of this  page,
	\begin{figure*}[ht]
		\begin{align}\label{UBRE1}
			&{\log _2}\left( { {1 + {\rho _S}\sqrt {\frac{{4\beta _0^3{K^2}}}{{{{\tilde d}_{SE}}\left[ n \right]{{\tilde d}_{RE}}\left[ n \right]{d_{SR}}}}}  + \frac{{{\rho _S}\beta _0^2{K^2}}}{{{d_{SR}}{{\tilde d}_{RE}}\left[ n \right]}} + \frac{{{\beta _0}{K^2}}}{{{{\tilde d}_{RE}}\left[ n \right]}} + }\frac{{{\rho _S}{\beta _0}}}{{{{\tilde d}_{SE}}\left[ n \right]}}} \right) \nonumber\\
			&\ge {\log _2}{A_4}\left[ n \right] + \frac{{{B_4}\left[ n \right]}}{{\ln 2{A_4}\left[ n \right]}}\left( {{z_h}\left[ n \right] - z_h^{\left( j \right)}\left[ n \right]} \right) +  \frac{{{C_4}\left[ n \right]}}{{\ln 2{A_4}\left[ n \right]}}\left( {{z_{ER}}\left[ n \right] - z_{ER}^{\left( j \right)}\left[ n \right]} \right) \buildrel \Delta \over = \varsigma _{RE1}^{{\mathrm{lb}}}\left[ n \right]
		\end{align}
		\hrulefill
	\end{figure*}
	where 
	${A_4}\left[ n \right] = 1 + \frac{{{\beta _0}{K^2}}}{{{\tilde d}_{RE}^{\left( j \right)}\left[ n \right]}} + \frac{{{\rho _S}{\beta _0}}}{{{\tilde d}_{SE}^{\left( j \right)}\left[ n \right]}} + \frac{{{\rho _S}\beta _0^2{K^2}}}{{{d_{SR}}{{\tilde d}_{RE}^{\left( j \right)}\left[ n \right]}}} + {\rho _S}\sqrt {\frac{{4\beta _0^3{K^2}}}{{{d_{SR}}{{\tilde d}_{SE}^{\left( j \right)}\left[ n \right]}{{\tilde d}_{RE}^{\left( j \right)}\left[ n \right]}}}}$,
	${B_4}\left[ n \right] =  - \frac{{{\beta _0}{K^2}}}{{{{\left( {{\tilde d}_{RE}^{\left( j \right)}\left[ n \right]} \right)}^2}}} - \frac{{\beta _0^2{K^2}{\rho _S}}}{{{d_{SR}}{{\left( {{\tilde d}_{RE}^{\left( j \right)}\left[ n \right]} \right)}^2}}} - {\rho _S}\sqrt {\frac{{4\beta _0^3{K^2}}}{{{{\tilde d}_{SE}^{\left( j \right)}\left[ n \right]}{d_{SR}}{{\left( {{\tilde d}_{RE}^{\left( j \right)}\left[ n \right]} \right)}^3}}}}$, 
	${C_4}\left[ n \right] =  - \frac{{{\beta _0}{\rho _S}}}{{{{\left( {{\tilde d}_{SE}^{\left( j \right)}\left[ n \right]} \right)}^2}}} - {\rho _S}\sqrt {\frac{{4\beta _0^3{K^2}}}{{{d_{SR}}{{\left( {{\tilde d}_{SE}^{\left( j \right)}\left[ n \right]} \right)}^3}{{\tilde d}_{RE}^{\left( j \right)}\left[ n \right]}}}} $, 
	$z_{ER}^{\left( j \right)}$ and $z_h^{\left( j \right)}$ denote the $j$-th iteration of $z_{ER}$ and $z_h$, 
	${{\tilde d}_{SE}^{\left( j \right)}\left[ n \right]}$, ${{\tilde d}_{ED}^{\left( j \right)}\left[ n \right]}$, ${{\tilde d}_{RE}^{\left( j \right)}\left[ n \right]}$ represent the $j$-th iteration values of ${{{\tilde d}_{SE}}}$, ${{{\tilde d}_{ED}}}$, ${{{\tilde d}_{RE}}}$, respectively.
	
	Similarly, the lower-bound of $z_E^2\left[ n \right]$ is obtained as
	\begin{align}\label{ZE2}
		z_E^2\left[ n \right] &\ge {\left({z_E^{\left( j \right)}}\left[ n \right]\right)^2} + 2{z_E^{\left( j \right)}}\left[ n \right]\left({{z_E}\left[ n \right]} - {z_E^{\left( j \right)}}\left[ n \right]\right) \nonumber\\
		&\buildrel \Delta \over =  \varsigma _H^\mathrm{lb}\left[ n \right].
	\end{align}
	
	Similar to (\ref{P20b}), the SCA is utilized to deal with ${\tilde R}_D\left[ n \right]$ and ${{\tilde R}_{E2}}\left[ n \right]$, which are given as 
	\begin{align}
		{\tilde R}_D\left[ n \right] &\le {\log _2}{A_5}\left[ n \right] + \frac{{{B_5}\left[ n \right]}}{{\ln 2{A_5}\left[ n \right]}}\left({z_h}\left[ n \right]-z_h^{\left( j \right)}\left[ n \right]\right) \nonumber\\
		& \buildrel \Delta \over =  \varsigma _{RD}^\mathrm{ub}\left[ n \right]
	\end{align}
	and 
	\begin{align}	\label{tildRE2}
		{{\tilde R}_{E2}}\left[ n \right] &= {\log _2}\left(1 + \frac{{{\beta _0}{K^2}}}{{{{\left\| {{{\mathbf{q}}_E}\left[ n \right] - {{\mathbf{q}}_R}} \right\|}^2} + {z_{ER}\left[ n \right]}}}\right) \nonumber\\
		&\ge {\log _2}{A_6}\left[ n \right] + \frac{{{B_6}\left[ n \right]}}{{\ln 2{A_6}\left[ n \right]}}\left({z_{ER}}\left[ n \right] - z_{ER}^{\left( j \right)}\left[ n \right]\right) \nonumber\\
		& \buildrel \Delta \over = \varsigma _{RE2}^\mathrm{lb}\left[ n \right],
	\end{align}
	where
	${A_5}\left[ n \right] = 1 + \frac{{{K^2}h_{SR}^2 h_{RD}^2{P_S}}}{{\left( {1 + {K^2}h_{RD}^2} \right){\sigma ^2} + \frac{{{P_E}\left[ n \right]{\beta _0}}}{{\tilde d_{ED}^{\left( j \right)}\left[ n \right]}}}}$, 
	${B_5}\left[ n \right] = \frac{{{K^2}h_{SR}^2 h_{RD}^2{P_S}\frac{{{P_E}\left[ n \right]{\beta _0}}}{{{{(\tilde d_{ED}^{\left( j \right)}\left[ n \right])}^2}}}}}{{{{\left( {\left( {1 + {K^2}h_{RD}^2} \right){\sigma ^2} + \frac{{{P_E}\left[ n \right]{\beta _0}}}{{\tilde d_{ED}^{\left( j \right)}\left[ n \right]}}} \right)}^2}}}$, 
	${A_6}\left[ n \right] = 1 + \frac{{{\beta _0}{K^2}}}{{{{\left\| {{{\mathbf{q}}_E}\left[ n \right] - {{\mathbf{q}}_R}} \right\|}^2} + z_{ER}^{\left( j \right)}\left[ n \right]}}$, 
	${B_6}\left[ n \right] =  - \frac{{{\beta _0}{K^2}}}{{{{\left({{\left\| {{{\mathbf{q}}_E}\left[ n \right] - {{\mathbf{q}}_R}} \right\|}^2} + z_{ER}^{\left( j \right)}\left[ n \right]\right)}^2}}}$,
	and 
	$\tilde d_{ED}^{\left( j \right)}$ denotes the $j$th iteration values of ${{{\tilde d}_{ED}}}$.
	
	Then, $\mathcal{P}_{2.2}$ is transformed as
	\begin{subequations}
		\begin{align}
			\mathcal{P}_{2.3} \,:\,&\mathop {\max }\limits_{{\mathbf{Z}},\mathbf{S_2},\mathbf{Z_1}, {v}_{z}, {a}_{z}} {\Theta _2} \\
			{\mathrm{s.t.}}\; 	&S_{4}\left[ n \right] \le \varsigma _{RD}^\mathrm{ub}\left[ n \right],\\
			&S_{5}\left[ n \right] \le \varsigma _{RE1}^\mathrm{lb}\left[ n \right],\\
			&S_{6}\left[ n \right] \ge \varsigma _{RE2}^\mathrm{lb}\left[ n \right],\\
			&S_{4}\left[ n \right] \le S_{5}\left[ n \right] - S_{6}\left[ n \right],\\
			&{z_h}\left[ n \right] \le \varsigma _H^\mathrm{lb}\left[ n \right],\\
			&{z_{ER}}\left[ n \right] \ge {\left({{z_E}\left[ n \right]} - {z_R}\right)^2},\\
			& (\textrm{\ref{p2b}})-(\textrm{\ref{p2h}}). \notag 
		\end{align}
	\end{subequations}
	The above problem is a standard convex problem, which can be solved by CVX.

	\begin{algorithm}[t]
		\caption{An iterative algorithm for joint optimization of 3D trajectory and jamming power}
		\begin{algorithmic}[1]  
			\STATE \textbf{Input}: ${{\mathbf{q}}_S}$, ${{\mathbf{q}}_D}$, $\mathbf{q}_I$, $\mathbf{q}_F$, $v^{\max}_{xy}$, $a^{\max}_{xy}$, $v^{\max}_{z}$, $a^{\max}_{z}$,  $H_{\max}$, $H_{\min}$,
			\STATE \textbf{Output:}$\mathbf{P}$, $\mathbf{Q}$, $\mathbf{Z}$
			\STATE  Initialization: set initial variables; tolerance $\varepsilon > 0$ and iteration number $j=0$.
			\STATE \textbf{Repeat}    
			\STATE With given $\mathbf{P}^{\left( {j} \right)}$, $ \mathbf{Z}^{{\left( {j} \right)}}$, obtain $\mathbf{Q}^{{\left( {j + 1} \right)}}$, $\mathbf{S_1}^{{\left( {j + 1} \right)}}$, $\mathbf{d}^{{\left( {j + 1} \right)}}$ by solving problem  ($\mathcal{P}_{1.{\mathrm{1b}}}$)
			\STATE With given $\mathbf{P}^{\left( {j} \right)}$, $ \mathbf{Z}^{\left( {j} \right)}$, $\mathbf{Q}^{{\left( {j + 1} \right)}}$, $\mathbf{S_1}^{{\left( {j + 1} \right)}}$, $\mathbf{d}^{{\left( {j + 1} \right)}}$, obtain $\mathbf{Z}^{{\left( {j + 1} \right)}}$, $\mathbf{S_2}^{{\left( {j + 1} \right)}}$, $\mathbf{Z_1}^{{\left( {j + 1} \right)}}$ by solving problem  ($\mathcal{P}_{2.{{3}}}$)
			\STATE With given $\mathbf{Q}^{{\left( {j + 1} \right)}}$, $\mathbf{Z}^{{\left( {j + 1} \right)}}$, $\mathbf{S_2}^{{\left( {j + 1} \right)}}$, $\mathbf{Z_1}^{{\left( {j + 1} \right)}}$,  obtain $\mathbf{P}^{{\left( {j + 1} \right)}}$ by solving problem ($\mathcal{P}_{1.{\mathrm{2a}}}$)
			\STATE $j=j+1$
			\STATE \textbf{Until} {$\left| {{\Theta ^{\left( {j + 1} \right)}} - {\Theta ^{\left( j \right)}}} \right| \le \varepsilon $ or reaches the maximum number of iterations}
			\STATE Obtain the optimal solution
		\end{algorithmic} 
	\end{algorithm}
	
	\textbf{Algorithm 2} summarizes the details of overall iterations for $\mathcal{P}_{2}$. 
	The initial trajectory is set as a straight line, which is given by
	\begin{subequations}
		\begin{align}
			{{\mathbf{q}}_E}\left[ n \right]&=\mathbf{q}_I+\frac{n}{N+1}\left(\mathbf{q}_F-\mathbf{q}_I\right),n=0,1,...,N,\\
			{{z_E}\left[ n \right]} &= {H_{\mathrm{I}}} + \frac{n}{{N + 1}}\left( {{H_{\mathrm{F}}} - {H_{\mathrm{I}}}} \right),n =0, 1,...,N.
			\label{eq3.5} 
		\end{align}
	\end{subequations}

		\label{Convergence2}
		In this work, the ER maximization problem is decomposed into several subproblems. 
		The SCA method is utilized to solve the sub-problems and the BCD method is utilized to iteratively solve the optimization problem until the algorithm converges or the number of iterations reaches the maximum value. 
		Algorithm 1 has two sub-problems, and Algorithm 2 has three sub-problems. 
		The complexity of our proposed algorithm is $\mathcal{O}\left(N_{iter}\left(\sum\limits_{k = 1}^K {\left( {N_k^{3.5}} \right)}\right) \log \left( {\frac{1}{\varepsilon }} \right)\right)$, 
		where $N_{iter}$ is the number of iterations, $N_k$ represents the number of variables in the subproblems, 
		and $K$ is the number of subproblems \cite{LiuT2021TGCN,LiangY2022TVT} .

	\begin{figure*}[t]
		\centering
		\subfigure[The flying trajectory of $E$.]{
			\label{fig2a}
			\includegraphics[width = 0.23 \textwidth]{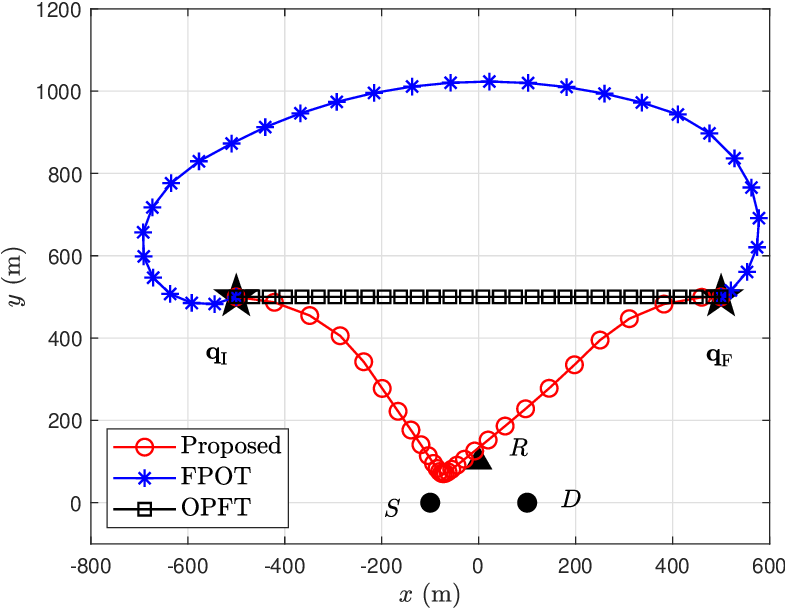}}
		\subfigure[The flying speed of $E$.]{
			\label{fig2b}
			\includegraphics[width = 0.23 \textwidth]{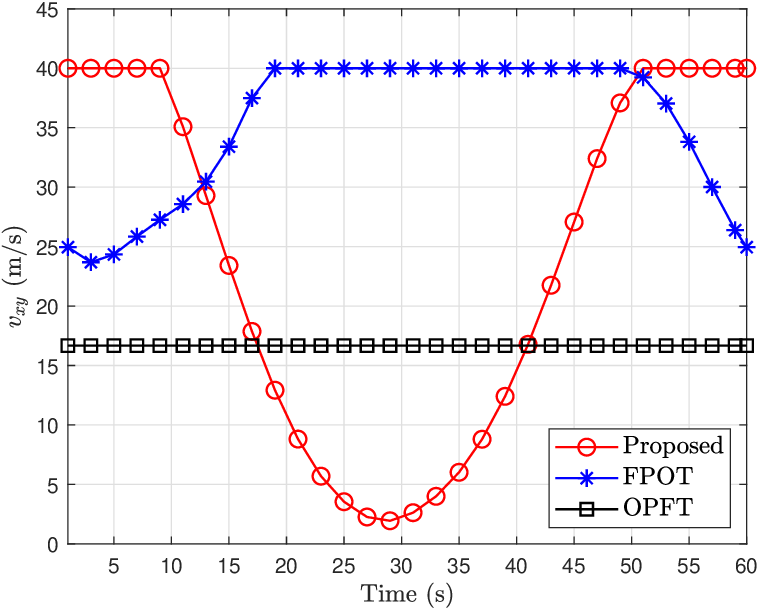}}	
		\subfigure[The jamming and flying power of $E$.]{
			\label{fig2c}
			\includegraphics[width = 0.23 \textwidth]{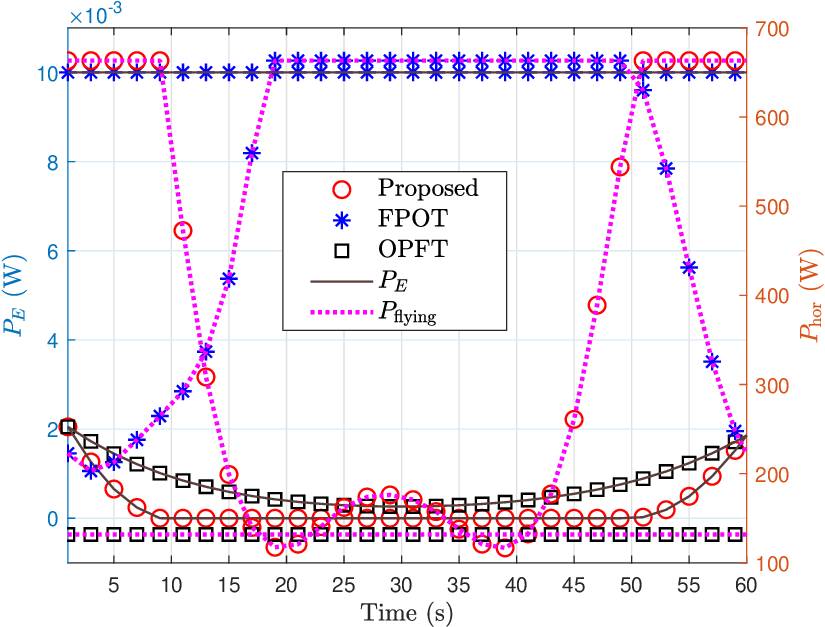}}
		\subfigure[The ER.]{
			\label{fig2d}
			\includegraphics[width = 0.23 \textwidth]{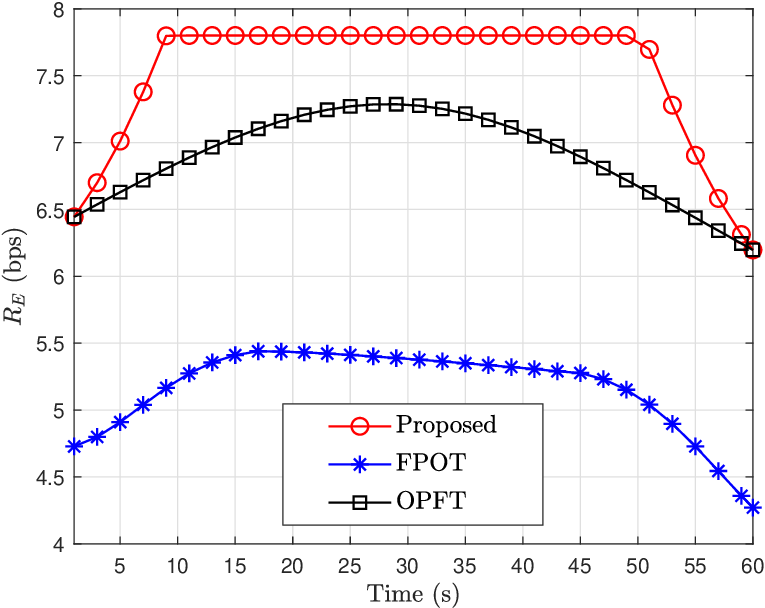}}
		\caption{Scenario 1.}
		\label{fig2}
	\end{figure*}
	\begin{figure*}[t]
		\centering
		\subfigure[The flying trajectory of $E$.]{
			\label{fig3a}
			\includegraphics[width = 0.23 \textwidth]{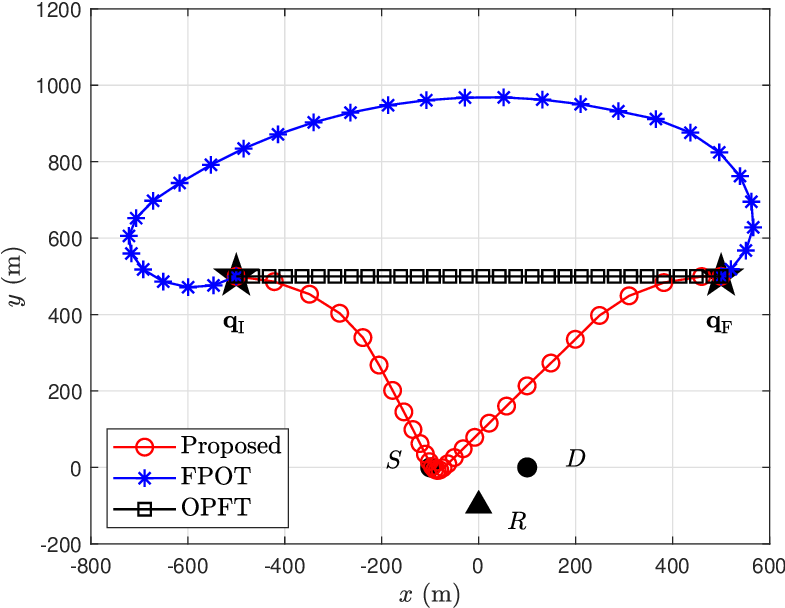}}
		\subfigure[The flying speed of $E$.]{
			\label{fig3b}
			\includegraphics[width = 0.23 \textwidth]{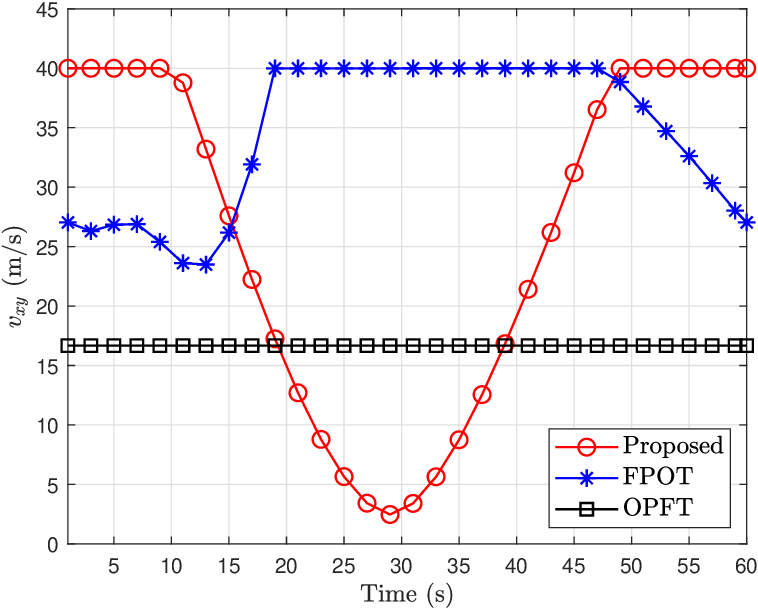}}
		\subfigure[The jamming and flying power of $E$.]{
			\label{fig3c}
			\includegraphics[width = 0.23 \textwidth]{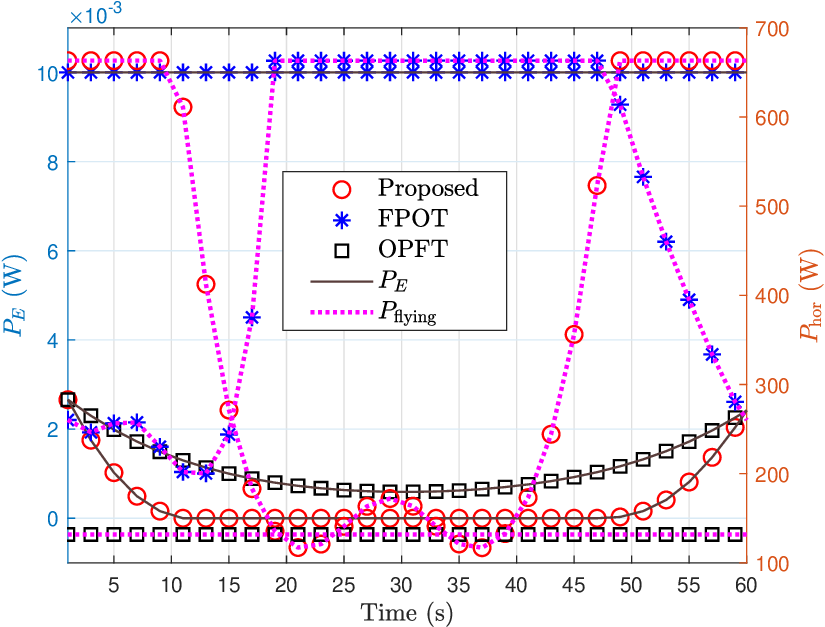}}
		\subfigure[The ER.]{
			\label{fig3d}
			\includegraphics[width = 0.23 \textwidth]{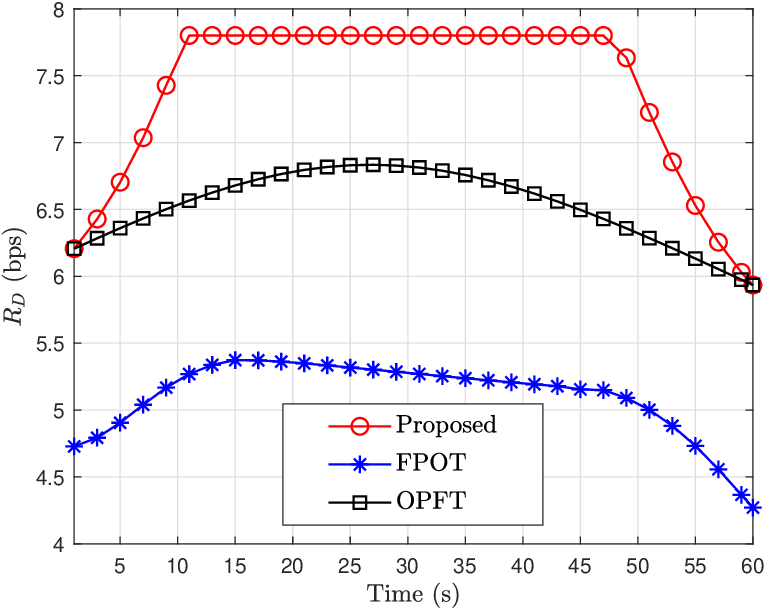}}
		\caption{Scenario 2.}
		\label{fig3}
	\end{figure*}

	\section{Numerical Results and Discussion}
	\label{sec:Simulation}

	In this section, we show numerical results to demonstrate the effectiveness of our proposed algorithm.
	The details of the simulation parameters are given in Table \ref{table}.
	The following schemes are utilized as benchmarks to verify the effectiveness of our proposed algorithm. 
	1) Benchmark 1: with the given fixed jamming power, the flight trajectory of $E$ is optimized {(denoted as FPOT)}.
	2) Benchmark 2: with the given fixed flight trajectory, the jamming power of $E$ is optimized {(denoted as OPFT)}.
	
	\begin{table}[tb]
			\begin{center}
				\caption{\textit{Simulation Parameters}}
				\begin{tabular}{c| c| c| c} 
					\Xhline{1.2pt}
					Parameters    &Value & Parameters  &   Value   \\ 
					\hline
					${\mathbf{q}}_I$ &  	${\left[ { -500,500} \right]^T}$&${\mathbf{q}}_F$ &  ${\left[500,500\right]^T}$\\
					\hline
					${\mathbf{q}}_S $ &${\left[-100,0\right]^T}$&${\mathbf{q}}_R$ & \makecell[c]{${\left[ { 0,100} \right]^T}$,\,\, 	${\left[ { 0,-100} \right]^T}$}\\
					\hline
					${\mathbf{q}}_D$ & ${\left[100,0\right]^T}$ & ${{H}}_I$ & 100 m/200 m 	\\		
					\hline
					${{H}}_F$ & 200 m &	${{z}}_R$ & 50 m	\\				
					\hline
					$v^{\max}_{xy}$	& 40 m/s &	$v^{\max}_{z}$	&	30 m/s 				\\ 
					\hline
					$a^{\max}_{xy}$	& 5 m/${\mathrm{s}}^2$  &$a^{\max}_{z}$	&  3 m/${\mathrm{s}}^2$\\
					\hline
					$H_{max}$					& 200 m	&$H_{min}$				& 60 m\\
					\hline
					$ \sigma^2$				& -110 dBm&	$\beta _0$				& -50 dB			\\ 
					\hline
					$P_S = P_R = P$ 					& 10 dBm&$\varepsilon$  & 0.001     \\
					\hline
					${P_0}$&   59 w&${P_i}$&    124 w\\
					\hline
					${v_0}$&   4.03&${U_{\mathrm{tip}}}$&     200 m/s \\
					\hline
					${d_0}$&        0.6&	$\rho$ &         1.225 $kg/m^3$\\
					\hline
					$s$ &           0.05&$A$ &           0.503 $m^2$\\
					\hline
					$P_{\mathrm{hor}}^{\mathrm{ave} }$& 600 w&	$P_{\mathrm{ver}}^{\mathrm{ave} }$& 300 w\\
					\hline
					$W$ & 20 N &&\\
					\Xhline{1.2pt}
				\end{tabular}
				\label{table}
			\end{center}
	\end{table}
	

		$E$ receives signals from both from $S$ and $R$. To illustrate the effect of the two-hop signals on $E$, two scenarios with different locations of $R$ are considered.
	Figs. \ref{fig2} and \ref{fig3} plot the flight trajectory, flight speed, jamming power, and ER of the three schemes in two scenarios with different $R$ positions.
	From Figs. \ref{fig2a} and \ref{fig3a}, it can be seen that $E$ in the proposed scheme first flies to the $S$ and $R$, then hovers for some time and flies back to the endpoint. 
	The trajectory designed by the proposed scheme is much closer to $S$ and $R$ than the benchmarks.
	Based on Figs. \ref{fig2b} and \ref{fig3b}, we can observe that, at first, $E$ flies towards $S$ and $R$ at its maximum speed, then gradually decreases the speed close to $S$ and $R$, then gradually increases the speed away, and finally flies towards the end at its maximum speed.
	From Figs. \ref{fig2c} and \ref{fig3c}, it is obvious that, the closer to $S$ and $R$, the smaller the jamming power is. 
	{{
			This is the best scenario to maximize ER by remaining $R_D$ without AN as it is. Due to the change in flight speed, $E$ has higher flight power when it is far away from $S$ and $R$ and lower flight power when it is near $S$ and $R$.	}
		}	Figs. \ref{fig2d} and \ref{fig3d} demonstrate that the ER of the proposed algorithm outperforms that of benchmarks. 
	This indicates that optimizing both the flying trajectory and jamming power simultaneously is the optimal solution.
	Fig. \ref{fig3} shows the results in the scenario with different $R$'s locations. Since $R$ is farther from the starting point of $E$, the speed in Fig. \ref{fig3b} rises and decreases more slowly than that in Fig. \ref{fig2b}, and the initial jamming power in  Fig. \ref{fig3c} is higher than that in  Fig. \ref{fig2c}. The initial ER in  Fig. \ref{fig3d} is also lower than in  Fig. \ref{fig2d}. 
	
	\begin{figure*}[t]
		\centering
		\subfigure[$\bar R_D$ for varying $P$.]{
			\label{fig4a}
			\includegraphics[width = 0.4  \textwidth]{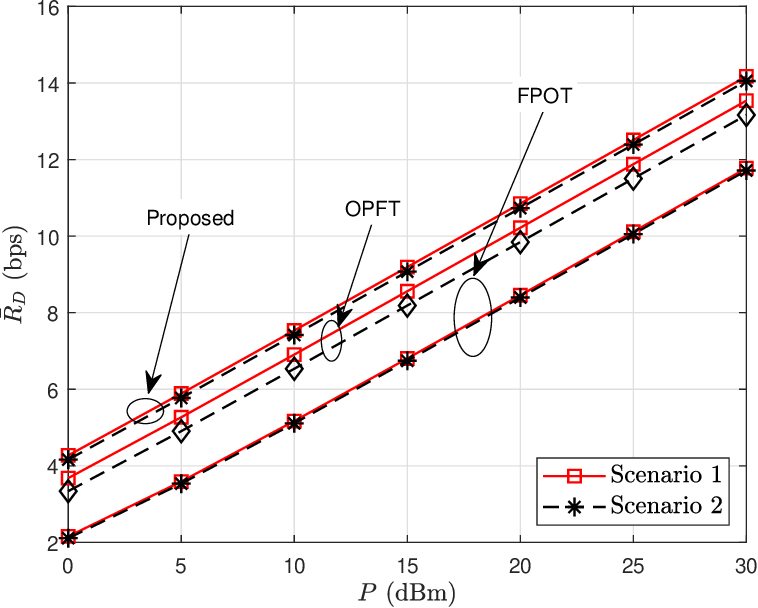}}
		\subfigure[$\bar R_D$ for varying $T$.]{
			\label{fig4b}
			\includegraphics[width = 0.4  \textwidth]{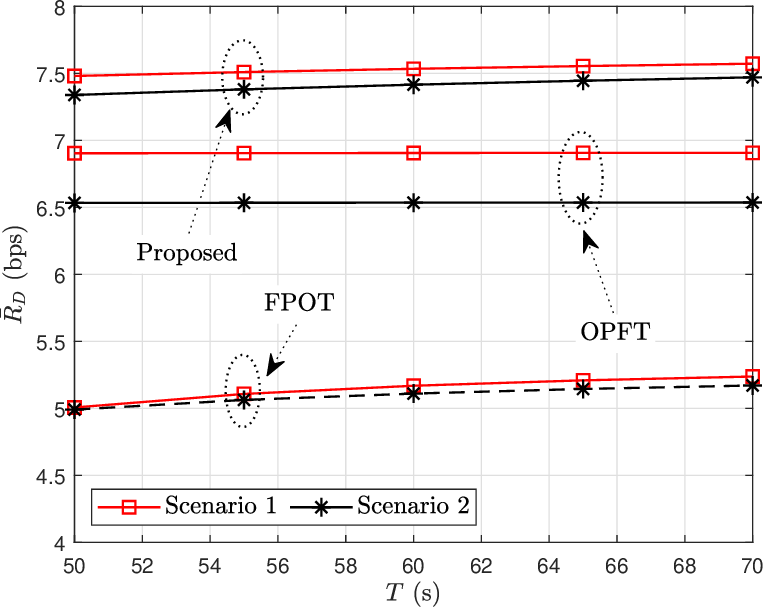}}
		\caption{AER for varying $P$ and $T$.}
		\label{fig4}
	\end{figure*}

Fig. \ref{fig4} describes the AER of all the schemes for varying transmitting power $P$ and flight period $T$. 
In Fig. \ref{fig4a}, the AER increases as $P$ increases. 
It can be seen that the AER gradually increases with the increase of $P$ in Fig. \ref{fig4a}. 
The effect of the flight period on the ER is not apparent, which can be observed from Fig. \ref{fig4b}.
The results in Fig. \ref{fig4} prove that the ER of the proposed scheme is the largest, while that of FPOT is the smallest, indicating that optimizing the jamming power is more effective than optimizing the trajectory. The jamming power and trajectory are optimized at the same time.
Comparing the results in different scenarios, it can be found that the position of $R$ has a more significant impact on OPFT.

\begin{figure}[t]
	\centering
	\includegraphics[width = 3 in]{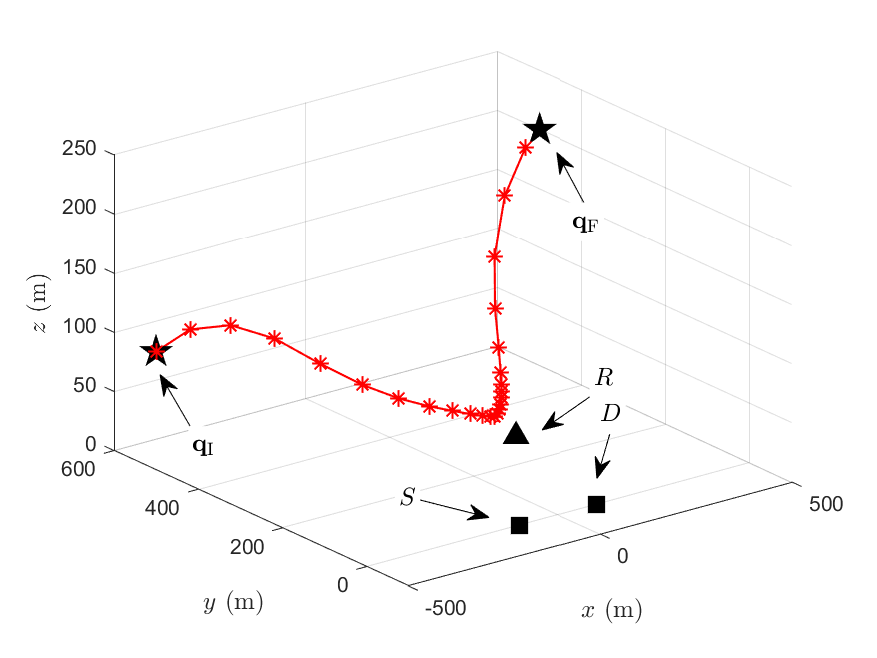}
	\caption{The 3D trajectory of $E$.}
	\label{fig5}
\end{figure}

Fig. \ref{fig5} plots the trajectories in scenarios where the initial and final points of $E$ have different altitudes. 	
Fig. \ref{fig6} shows the projection trajectory, the altitude, horizontal speed, and vertical speed of $E$.
In Fig. \ref{fig6a}, it can be found that the trajectory is basically the same as that in Fig. \ref{fig2}. 
Specifically, $E$ gradually approaches $S$ and $R$ from the starting point, hovers around $S$ and $R$ for some slots, and finally returns to the endpoint.
In Fig. \ref{fig6b}, we can see that $E$ first raises the flight altitude, then descends after some time and gradually approaches $S$ and $R$, then stabilizes at the altitude for some time, lowers some altitude, and then quickly raises the altitude to return to the endpoint. 
Moreover, as the flight period increases, $E$'s flight altitude does not decrease significantly.
The horizontal and vertical velocities of the 3D trajectory are given in Figs. \ref{fig6c} and \ref{fig6d}, respectively.
In Fig. \ref{fig6c}, $E$ gradually decreases after a period of maximum horizontal speed, then gradually increases its speed at a low speed for a while, and finally returns to the endpoint at maximum level.
In Fig. \ref{fig6d}, the vertical velocity of $E$ decreases first and then rises. After maintaining the vertical velocity near 0 m/s for a while, the vertical velocity decreases again, then gradually rises, and finally falls again.

	\begin{figure*}[t]
		\centering
		\subfigure[The 2D projection of $E$.]{
			\label{fig6a}
			\includegraphics[width = 0.23 \textwidth]{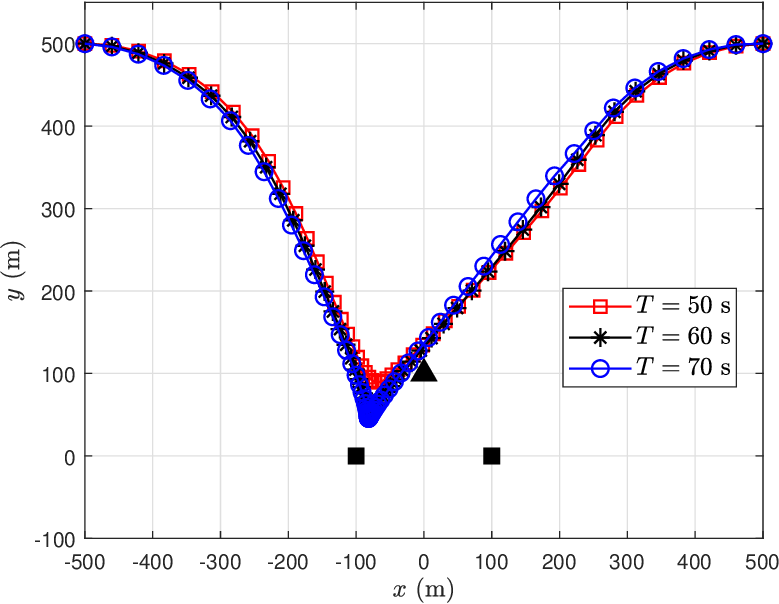}}
		\subfigure[The height of $E$.]{
			\label{fig6b}
			\includegraphics[width = 0.23 \textwidth]{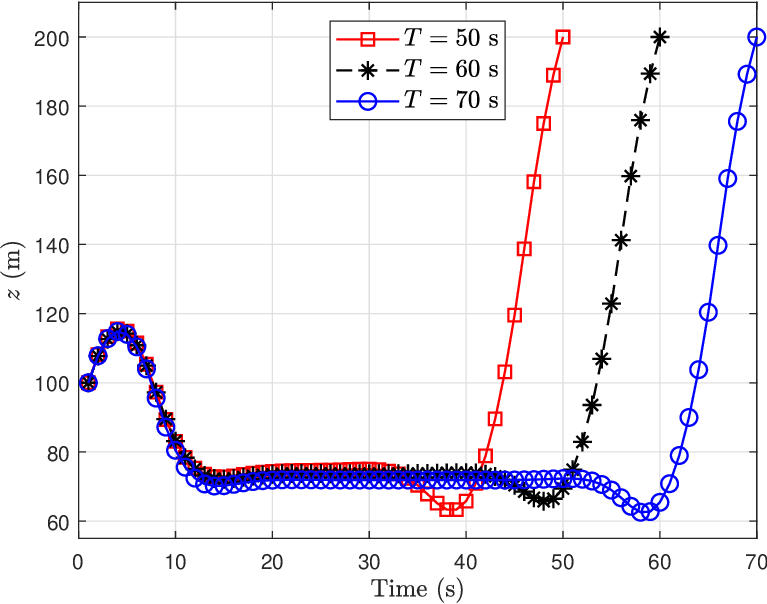}}	
		\subfigure[The horizontal velocity of $E$.]{
			\label{fig6c}
			\includegraphics[width = 0.23 \textwidth]{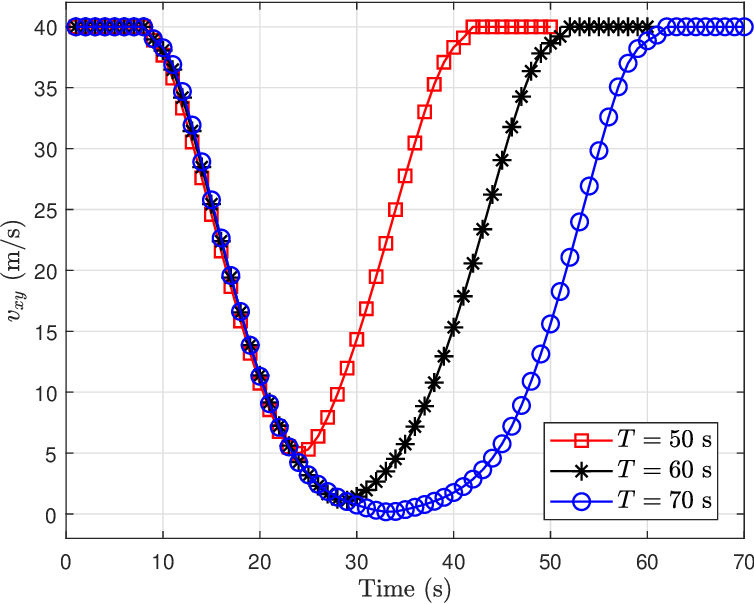}}
		\subfigure[The vertical velocity of $E$.]{
			\label{fig6d}
			\includegraphics[width = 0.23 \textwidth]{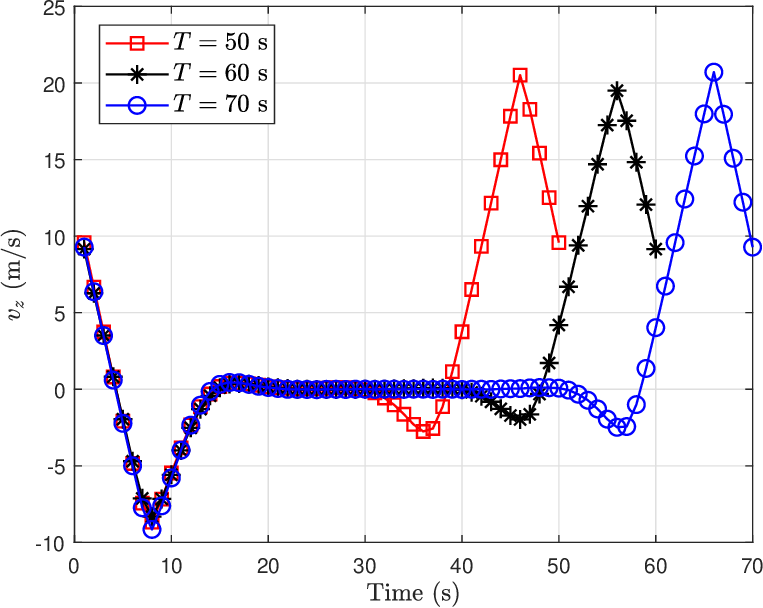}}
		\caption{3D flight trajectories of $E$.}
		\label{fig6}
	\end{figure*}

	\begin{figure*}[t]
		\centering
		\subfigure[Jamming and flying power for varying $T$.]{
			\label{fig7a}
			\includegraphics[width = 0.4  \textwidth]{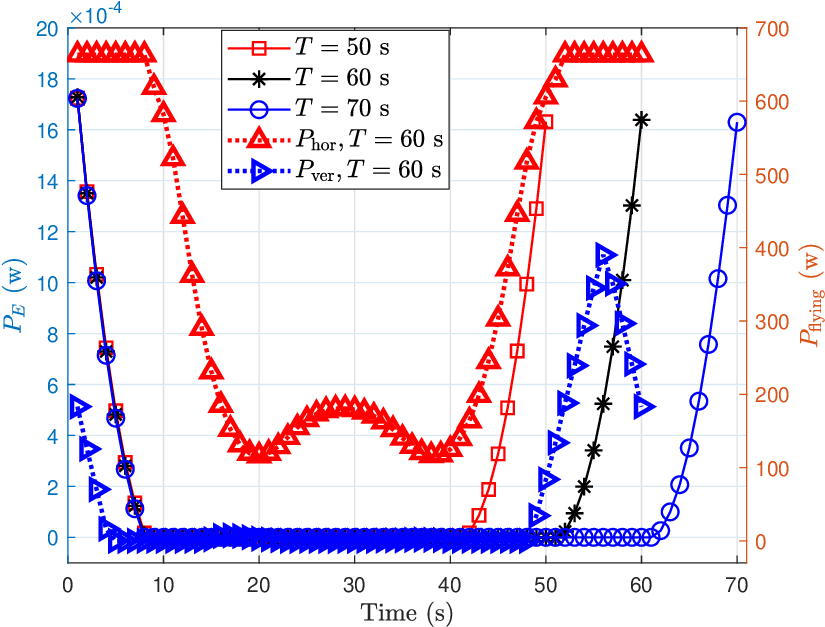}}
		\subfigure[ER for varying $T$.]{
			\label{fig7b}
			\includegraphics[width = 0.4  \textwidth]{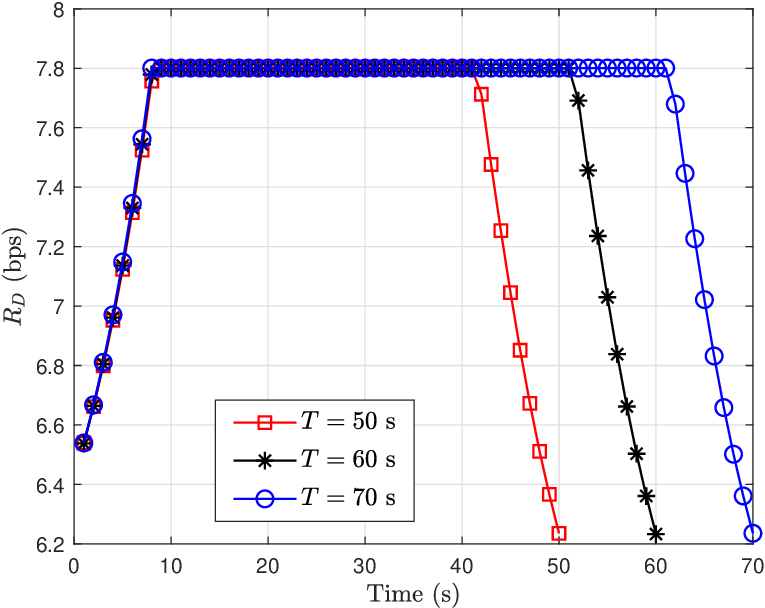}}
		\caption{Jamming power, flying power, and ER of UAV for varying $T$.}
		\label{fig7}
	\end{figure*}

	{Fig. \ref{fig7} shows the jamming power, flight power, and ER of $E$ for varying flight periods. 
	From Fig. \ref{fig7a}, it can be seen that, at the starting point, the jamming power is considerable enough to interfere with $D$. As $E$ approaches $S$ and $R$, the jamming power gradually decreases until it reaches zero, the horizontal flying power is gradually reduced, and the vertical flight power is also close to zero.} 
	Fig. \ref{fig7b} shows that the ER gradually increases. The longer the running time, the more time will be available for interference-free eavesdropping.

	\begin{figure}[t]
		\centering
		\includegraphics[width = 3 in]{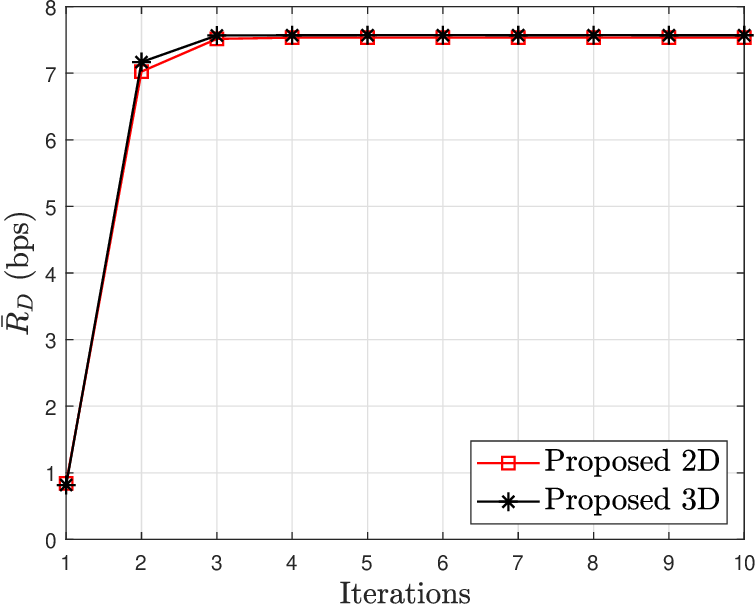}
		\caption{Algorithm convergence.}
		\label{fig8}
	\end{figure}

	Fig. \ref{fig8} illustrates the convergence of our proposed iterative algorithm. 
	It is observed that the AER increases quickly within a small number of iterations and rapidly converges within around four iterations. 
	It can be seen that the iterative algorithm's AER tends to be constant when iterating four times, which means the algorithm converges fast.

\section{Conclusion}
\label{sec:Conclusions}
	
This work investigated the eavesdropping performance of an illegitimate cooperative system with a proactive eavesdropper UAV. 
The jamming power and trajectory of the UAV were jointly considered to maximize the ER. 
An iterative algorithm based on successive convex approximation and block coordinate descent method was proposed to solve this problem. 
Simulation results show the effectiveness of our proposed algorithm. 
{ 
	In this work, the scenarios with single-antenna were considered. The proactive eavesdropping problem of an illegitimate multi-antenna cooperative system with a proactive eavesdropper UAV is an exciting topic that will be addressed in future work. 
}		
		
	\begin{appendices}
		
		\section{Proof of Lemma 1 }
		\label{sec:appendicesA}	
		
		The first-order partial derivation of $f$ with respect to $x$ and $y$ are obtained as
		\begin{align}\label{derfx} 
			\frac{{\partial f}}{{\partial x}} &= \frac{{ - {c_2}{x^{ - 2}} - \frac{1}{2}{c_4}{x^{ - \frac{3}{2}}}{y^{ - \frac{1}{2}}}}}{{\ln 2({c_1} + {c_2}{x^{ - 1}} + {c_3}{y^{ - 1}} + {c_4}{x^{ - \frac{1}{2}}}{y^{ - \frac{1}{2}}})}}
		\end{align}
		and
		\begin{align}\label{derfy} 
			\frac{{\partial f}}{{\partial y}} &= \frac{{ - {c_3}{y^{ - 2}} - \frac{1}{2}{c_4}{x^{ - \frac{1}{2}}}{y^{ - \frac{3}{2}}}}}{{\ln 2({c_1} + {c_2}{x^{ - 1}} + {c_3}{y^{ - 1}} + {c_4}{x^{ - \frac{1}{2}}}{y^{ - \frac{1}{2}}})}},
		\end{align}
		respectively.
		
		The second-order partial derivation of $f$ with respect to $x$ and $y$ are obtained as
		\begin{align}\label{derfx2} 
			\frac{{{\partial ^2}f}}{{\partial {x^2}}} &= \eta (2{c_1}{c_2}{x^{ - 3}} + {c_2}{c_2}{x^{ - 4}} + 2{c_2}{c_3}{x^{ - 3}}{y^{ - 1}}  \notag \\
			& + {c_2}{c_4}{x^{ - \frac{7}{2}}}{y^{ - \frac{1}{2}}} + \frac{3}{4}{c_1}{c_4}{x^{ - \frac{5}{2}}}{y^{ - \frac{1}{2}}} + \frac{3}{4}{c_2}{c_4}{x^{ - \frac{7}{2}}}{y^{ - \frac{1}{2}}}  \notag \\&+ \frac{3}{4}{c_3}{c_4}{x^{ - \frac{5}{2}}}{y^{ - \frac{3}{2}}} + \frac{2}{4}{c_4}{c_4}{x^{ - \frac{6}{2}}}{y^{ - \frac{2}{2}}}),
		\end{align}
		\begin{align}\label{derfxy}
			\frac{{{\partial ^2}f}}{{\partial x\partial y}} &= \eta (\frac{1}{4}{c_1}{c_4}{x^{ - \frac{3}{2}}}{y^{ - \frac{3}{2}}} - {c_2}{c_3}{x^{ - 2}}{y^{ - 2}}  \notag \\
			& - \frac{1}{4}{c_2}{c_4}{x^{ - \frac{5}{2}}}{y^{ - \frac{3}{2}}} - \frac{1}{4}{c_3}{c_4}{x^{ - \frac{3}{2}}}{y^{ - \frac{5}{2}}}),
		\end{align}
		and
		\begin{align}\label{derfy2} 
			\frac{{{\partial ^2}f}}{{\partial {y^2}}} &=\eta (2{c_1}{c_3}{y^{ - 3}} + 2{c_2}{c_3}{x^{ - 1}}{y^{ - 3}} + {c_3}{c_3}{y^{ - 4}}  \notag\\
			&+ {c_3}{c_4}{x^{ - \frac{1}{2}}}{y^{ - \frac{7}{2}}} + \frac{3}{4}{c_1}{c_4}{x^{ - \frac{1}{2}}}{y^{ - \frac{5}{2}}} + \frac{3}{4}{c_2}{c_4}{x^{ - \frac{3}{2}}}{y^{ - \frac{5}{2}}} \notag \\
			& + \frac{3}{4}{c_3}{c_4}{x^{ - \frac{1}{2}}}{y^{ - \frac{7}{2}}} + \frac{2}{4}{c_4}{c_4}{x^{ - \frac{2}{2}}}{y^{ - \frac{6}{2}}}),
		\end{align}
		where
		$\eta = {{{{(\ln 2({c_1} + {c_2}{x^{ - 1}} + {c_3}{y^{ - 1}} + {c_4}{x^{ - \frac{1}{2}}}{y^{ - \frac{1}{2}}}))}^{-2}}}}$.
		
		The Hessian matrix of $f$ is expressed as
		\begin{align}\label{Hessian}
			{\nabla ^2}f &= \left( {\begin{array}{*{20}{c}}
					{\frac{{{\partial ^2}f}}{{\partial {x^2}}}}&{\frac{{{\partial ^2}f}}{{\partial x\partial y}}},\\
					{\frac{{{\partial ^2}f}}{{\partial y\partial x}}}&{\frac{{{\partial ^2}f}}{{\partial {y^2}}}}
			\end{array}} \right) \nonumber\\
			&= \frac{{{\partial ^2}f}}{{\partial {x^2}}}\frac{{{\partial ^2}f}}{{\partial {y^2}}} - \frac{{{\partial ^2}f}}{{\partial x\partial y}}\frac{{{\partial ^2}f}}{{\partial y\partial x}} > 0,
		\end{align}
		
		Due to $\frac{{{\partial ^2}f}}{{\partial {x^2}}} > 0$ and ${\nabla ^2}f> 0$, $f$ is a convex function of $x$ and $y$.
		
	\end{appendices}


\begin{thebibliography}{1}
		
		\bibitem{ZengY2016CM}
		Y. Zeng, R. Zhang, and T. J. Lim, ``Wireless communications with unmanned aerial vehicles: Opportunities and challenges,'' \textit{IEEE Commun. Mag.}, vol. 54, no. 5, pp. 36-42, May 2016.
		
		\bibitem{Hayat2016COMST}
		S. Hayat, E. Yanmaz, and R. Muzaffar, ``Survey on unmanned aerial vehicle networks for civil applications: A communications viewpoint,'' \textit{IEEE Commun. Surveys Tuts.}, vol. 18, no. 4, pp. 2624-2661, 4th Quart. 2016.
		
		\bibitem{Bander2020JNCA}
		B. Alzahrani, O. S. Oubbati, A. Barnawi, M. Atiquzzaman, and D. Alghazzawi, ``UAV assistance paradigm: State-of-the-art in applications and challenges,'' \textit{J. Netw. Comput. Appl.}, vol. 166, Sep. 2020.
		
		\bibitem{Sahoo2022IOT}
		B. P. S. Sahoo, D. Puthal, and P. K. Sharma, ``Toward advanced UAV communications: Properties, research challenges, and future potential,'' \textit{IEEE Internet Things Mag.}, vol. 5, no. 1, pp. 154-159, Mar. 2022.
		
		\bibitem{ZhaoN2019WC}
		N. Zhao, W. Lu, M. Sheng, Y. F. Chen, J. Tang, F. R. Yu, and K. K. Wong, ``UAV-assisted emergency networks in disasters,'' \textit{IEEE Wireless Commun.}, vol. 26, no. 1, pp. 45-51, Feb. 2019.
		
		\bibitem{YangH2022IoT}
		H. Yang, R. Ruby, Q.-V. Pham, and K. Wu, ``Aiding a disaster spot via multi-UAV-based IoT networks: Energy and mission completion time-aware trajectory optimization,''\textit{ IEEE Internet Things J.}, vol. 9, no. 8, pp. 5853-5867, Apr. 2022.
		
		\bibitem{ZhangC2021TCOM}
		C. Zhang, L. Zhang, L. Zhu, T. Zhang, Z. Xiao, and X.-G. Xia, ``3D deployment of multiple UAV-mounted base stations for UAV communications,'' \textit{IEEE Trans. Commun.}, vol. 69, no. 4, pp. 2473-2488, Apr. 2021.
		
		\bibitem{LiY2021TWC}
		Y. Li, H. Zhang, K. Long, C. Jiang, and M. Guizani, ``Joint resource allocation and trajectory optimization with QoS in UAV-based NOMA wireless networks,'' \textit{IEEE Trans. Wireless Commun.}, vol. 20, no. 10, pp. 6343-6355, Oct. 2021.
		
		\bibitem{LiuT2021TGCN}
		T. Liu, M. Cui, G. Zhang, Q. Wu, X. Chu, and J. Zhang, ``3D trajectory and transmit power optimization for UAV-enabled multi-link relaying systems,'' \textit{IEEE Trans. Green Commun. Netw.}, vol. 5, no. 1, pp. 392-405, Mar. 2021.
		
		\bibitem{LiangY2022TVT}
		Y. Liang, L. Xiao, D. Yang, Y. Liu, and T. Zhang, ``Joint trajectory and resource optimization for UAV-aided two-way relay networks,'' \textit{IEEE Trans. Veh. Technol.}, vol. 71, no. 1, pp. 639-652, Jan. 2022.
		
		\bibitem{FengT2022TWC}
		T. Feng, L. Xie, J. Yao, and J. Xu, ``UAV-enabled data collection for wireless sensor networks with distributed beamforming,'' \textit{IEEE Trans. Wireless Commun.}, vol. 21, no. 2, pp. 1347-1361, Feb. 2022.
		
		\bibitem{Samir2020TWC}
		M. Samir, S. Sharafeddine, C. M. Assi, T. M. Nguyen, and A. Ghrayeb, ``UAV trajectory planning for data collection from time-constrained IoT device,'' \textit{IEEE Trans. Wireless Commun.}, vol. 19, no. 1, pp. 34-46, Jan. 2020.
		
		\bibitem{Dabiri2020CL}
		M. T. Dabiri and S. M. S. Sadough, ``Optimal placement of UAV-assisted free-space optical communication systems with DF relaying,'' \textit{IEEE Commun. Lett.}, vol. 24, no. 1, pp. 155-158, Jan. 2020.
		
		\bibitem{TangG2022SJ}
		G. Tang, P. Du, H. Lei, I. S. Ansari, and Y. Fu, ``Trajectory design and communication resources allocation for wireless powered secure UAV communication systems,'' \textit{IEEE Syst. J.}, vol. 16, no. 4, pp. 6300-6308, Dec. 2022.
		
		\bibitem{ZhongC2019CL}
		C. Zhong, J. Yao, and J. Xu, ``Secure UAV communication with cooperative jamming and trajectory control,'' \textit{IEEE Commun. Lett.}, vol. 23, no. 2, pp. 286-289, Feb. 2019.
		
		\bibitem{XuY2021TCOM}
		Y. Xu, T. Zhang, D. Yang, Y. Liu, and M. Tao, ``Joint resource and trajectory optimization for security in UAV-assisted MEC systems,'' \textit{IEEE Trans. Commun.}, vol. 69, no. 1, pp. 573-588, Jan. 2021.
		
		\bibitem{LeiH2023IoT}
		H. Lei, H. Yang, K.-H. Park, I. S. Ansari, J. Jiang, and M.-S. Alouini, ``Joint trajectory design and user scheduling for secure aerial underlay IoT systems," \textit{IEEE Internet Things J.}, vol. 10, no. 15, pp. 13637-13648, Aug. 2023.
		
		\bibitem{ZhouY2022TVT}
		Y. Zhou, P. L. Yeoh, C. Pan, K. Wang, Z. Ma, B. Vucetic, and Y. Li, ``Caching and UAV friendly jamming for secure communications with active eavesdropping attacks,'' \textit{IEEE Trans. Veh. Technol.}, vol. 71, no. 10, pp. 11251-11256, Oct. 2022.
		
		\bibitem{HanD2020ChinaCom}
		D. Han and T. Shi, ``Secrecy capacity maximization for a UAV-assisted MEC system,'' \textit{China Commun.}, vol. 17, no. 10, pp. 64-81, Oct. 2020.
		
		\bibitem{DuoB2021ChinaCom}
		B. Duo, J. Luo, Y. Li, H. Hu, and Z. Wang, ``Joint trajectory and power optimization for securing UAV communications against active eavesdropping,'' \textit{China Commun.}, vol. 18, no. 1, pp. 88-99, Jan. 2021.
		
		\bibitem{Savkin2020WC}
		A. V. Savkin, H. Huang, and W. Ni, ``Securing UAV communication in the presence of stationary or mobile eavesdroppers via online 3D trajectory planning," I\textit{EEE Wireless Commun. Lett.}, vol. 9, no. 8, pp. 1211-1215, Aug. 2020.
		
		\bibitem{WuJ2023TVT}
		J. Wu, W. Yuan, and L. Hanzo, ``When UAVs meet ISAC: Real-time trajectory design for secure communications,'' \textit{IEEE Trans. Veh. Technol.}, vol. 72, no. 12, pp. 16766-16771, Dec. 2023.
		
		\bibitem{ZengY2016TSP}
		Y. Zeng and R. Zhang, ``Wireless information surveillance via proactive eavesdropping with spoofing relay," \textit{IEEE J. Sel. Topics Signal Process.}, vol. 10, no. 8, pp. 1449-1461, Dec. 2016.
		
		\bibitem{XuJ2017TWC}
		 J. Xu, L. Duan, and R. Zhang, ``Proactive eavesdropping via cognitive jamming in fading channels," \textit{IEEE Trans. Wireless Commun.}, vol. 16, no. 5, pp. 2790-2806, May 2017.
		
		\bibitem{XuJ2017WC}
		J. Xu, L. Duan, and R. Zhang, ``Surveillance and intervention of infrastructure-free mobile communications: A new wireless security paradigm," \textit{IEEE Wireless Commun.}, vol. 24, no. 4, pp. 152-159, Aug. 2017.
		
		
		\bibitem{XuD2022TWC}
		D. Xu and H. Zhu, ``Proactive eavesdropping for wireless information surveillance under suspicious communication quality-of-service constraint,'' \textit{IEEE Trans. Wireless Commun.}, vol. 21, no. 7, pp. 5220-5234, Jul. 2022.
		
		\bibitem{XuD2023TIFS}
		D. Xu and H. Zhu, ``Proactive eavesdropping of physical layer security aided suspicious communications in fading channels,'' \textit{IEEE Trans. Inf. Forensics Security}, vol. 18, pp. 1111-1126, Jan. 2023.
		
		\bibitem{ZhangH2020TWC}
		H. Zhang, L. Duan, and R. Zhang, ``Jamming-assisted proactive eavesdropping over two suspicious communication links,'' \textit{IEEE Trans. Wireless Commun.}, vol. 19, no. 7, pp. 4817-4830, Jul. 2020.
		
		\bibitem{ChenJ2022TVT}
		J. Chen, L. Tang, D. Guo, Y. Bai, L. Yang, and Y.-C. Liang, ``Proactive eavesdropping in massive MIMO-OFDM systems via deep reinforcement learning,'' \textit{IEEE Trans. Veh. Technol.}, vol. 71, no. 11, pp. 12315-12320, Nov. 2022.
		
		\bibitem{Feizi2020TCOM}
		F. Feizi, M. Mohammadi, Z. Mobini, and C. Tellambura, ``Proactive eavesdropping via jamming in full-duplex multi-antenna systems: Beamforming design and antenna selection,'' \textit{IEEE Trans. Commun.}, vol. 68, no. 12, pp. 7563-7577, Dec. 2020.
		
		\bibitem{LuH2019TVT}
		H. Lu, H. Zhang, H. Dai, W. Wu, and B. Wang, ``Proactive eavesdropping in UAV-aided suspicious communication systems,'' \textit{IEEE Trans. Veh. Technol.}, vol. 68, no. 2, pp. 1993-1997, Feb. 2019.
		
		\bibitem{HuangM2021WCNCW}
		M. Huang, Y. Chen, and X. Tao, ``Proactive eavesdropping in UAV systems via trajectory planning and power optimization," in Proc. \textit{2021 IEEE Wireless Communications and Networking Conference Workshops (WCNCW)}, Nanjing, China, Mar. 2021, pp. 1-6.
		
		\bibitem{GuoD2023TMC}
		D. Guo, L. Tang, X. Zhang, and Y.-C. Liang, ``Joint optimization of trajectory and jamming power for multiple UAV-aided proactive eavesdropping,'' \textit{IEEE Trans. Mobile Comput.}, vol. 23, no. 5, pp. 5770-5785, May 2024
		
		\bibitem{DanQ2022phycom}
		Q. Dan, F. Wu, M. Xu, D. Yang, Y. Yu, Y. Zhang, ``Proactive eavesdropping rate maximization with UAV spoofing relay,'' \textit{Phys. Commun.}, vol. 55, Dec. 2022.
		
		
		\bibitem{JiangX2017SPL}
		X. Jiang, H. Lin, C. Zhong, X. Chen, and Z. Zhang, ``Proactive eavesdropping in relaying systems,'' \textit{IEEE Signal Process. Lett.}, vol. 24, no. 6, pp. 917-921, Jun. 2017.
		
		\bibitem{HuG2021SJ}
		G. Hu, J. Ouyang, Y. Cai, and Y. Cai, ``Proactive eavesdropping in two-way amplify-and-forward relay networks,'' \textit{IEEE Syst. J.}, vol. 15, no. 3, pp. 3415-3426, Sep. 2021.
		
		\bibitem{HuG2022SPL}
		G. Hu, J. Si, Y. Cai, and N. Al-Dhahir, ``Proactive eavesdropping via jamming in UAV-enabled relaying systems with statistical CSI,'' \textit{IEEE Signal Process. Lett.}, vol. 29, pp. 1267-1271, Jun. 2022.
		
		\bibitem{HuG2020CL}
		G. Hu and Y. Cai, ``Legitimate eavesdropping in UAV-based relaying system," \textit{IEEE Commun. Lett.}, vol. 24, no. 10, pp. 2275-2279, Oct. 2020.
		
		
		\bibitem{WuZ2023OPENJ}
		Z. Wu, K. Guo, X. Li, Y.-D. Zhang, M. Zhu, H. Song, and M. Wu, ``Proactive eavesdropping performance for integrated satellite–terrestrial relay networks,'' \textit{IEEE Open J. Commun. Soc.}, vol. 4, pp. 2985-2999, Nov. 2023.
		
		\bibitem{HuG2020CLAF}
		G. Hu and Y. Cai, ``UAVs-assisted proactive eavesdropping in AF multi-relay system," \textit{IEEE Commun. Lett.}, vol. 24, no. 3, pp. 501-505, Mar. 2020.
		
		\bibitem{Moon2018TWC}
		J. Moon, H. Lee, C. Song, S. Lee, and I. Lee, ``Proactive eavesdropping with full-duplex relay and cooperative jamming,'' \textit{IEEE Trans. Wireless Commun.}, vol. 17, no. 10, pp. 6707-6719, Oct. 2018.
		
		\bibitem{GeY2022IOT}
		Y. Ge and P. C. Ching, ``Energy efficiency for proactive eavesdropping in cooperative cognitive radio networks,'' \textit{IEEE Internet Things J.}, vol. 9, no. 15, pp. 13443-13457, Aug. 2022.
		
		\bibitem{WuQ2018TWC}	
		Q. Wu, Y. Zeng, and R. Zhang, ``Joint trajectory and communication design for multi-UAV enabled wireless networks," \emph{IEEE Trans. Wireless Commun.}, vol. 17, no. 3, pp. 2109-2121, Mar. 2018.
		
		\bibitem{CuiM2018TVT}
		M. Cui, G. Zhang, Q. Wu, and D. W. K. Ng, ``Robust trajectory and transmit power design for secure UAV communications,'' \textit{IEEE Trans. Veh. Technol.},vol. 67, no. 9, pp. 9042-9046, Sept. 2018.
		
		\bibitem{ZhangR2021TWC}
		R. Zhang, X. Pang, W. Lu, N. Zhao, Y. Chen, and D. Niyato, ``Dual-UAV enabled secure data collection with propulsion limitation,'' \textit{IEEE Trans. Wireless Commun.}, vol. 20, no. 11, pp. 7445-7459, Jun. 2021.
		
\bibitem{HuaM2020TCOM}
M. Hua, L. Yang, Q. Wu, and A. L. Swindlehurst, ``3D UAV trajectory and communication design for simultaneous uplink and downlink transmission,'' \textit{IEEE Trans. Veh. Technol.}, vol. 68, no. 9, pp. 5908-5923, Sept. 2020.

\bibitem{LeiH2023ArXivHaos}
H. Lei, J. Jiang, H. Yang, K.-H. Park, I. S. Ansari, G. Pan, and M.-S. Alouini. ``Trajectory and power design for aerial CRNs with colluding eavesdroppers," \textit{arXiv:2310.13931}, Oct. 2023, [Online]: https://arxiv.org/abs/2310.13931

\bibitem{ZengY2019TWC}
Y. Zeng, J. Xu, and R. Zhang, ``Energy minimization for wireless communication with rotary-wing UAV,'' \textit{IEEE Trans. Wireless Commun.}, vol. 18, no. 4, pp. 2329-2345, Apr. 2019.

\bibitem{FilipponeABook}
A. Filippone, \textit{Flight Performance of Fixed and Rotary Wing Aircraft.} Amsterdam, The Netherlands: Elsevier, 2006.
		
		
	\end{thebibliography}
\end{document}